\documentclass[12pt,oneside]{amsart}

\usepackage{amsmath, amssymb, amsthm, amscd, 
}

\usepackage[matrix,arrow,curve]{xy}

\usepackage{amsaddr}

\allowdisplaybreaks[4]

\newcommand{\eval}[2][\right]{\relax
  \ifx#1\right\relax \left.\fi#2#1\rvert}

\newcommand{\pd}{{\partial}}

\newcommand{\al}{{\alpha}}
\newcommand{\la}{{\lambda}}

\newcommand{\ost}{\mathbb{U}}

\newcommand{\er}{\eqref}
\newcommand{\cl}{\colon}

\newcommand{\beq}{\begin{equation}}
\newcommand{\ee}{\end{equation}}

\newcommand{\bmu}{\begin{multline}}
\newcommand{\emul}{\end{multline}}

\newcommand{\be}{{\beta}}
\newcommand{\anA}{\mathsf{A}}
\newcommand{\anB}{\mathsf{B}}
\newcommand{\anG}{\mathsf{G}}

\newcommand{\CE}{\mathcal{E}}
\newcommand{\ce}{\mathcal{E}}

\newcommand{\zp}{\mathbb{Z}_{\ge 0}}
\newcommand{\zsp}{\mathbb{Z}_{>0}}

\newcommand{\msl}{\mathfrak{sl}}
\newcommand{\gl}{\mathfrak{gl}}

\newcommand{\mg}{\mathfrak{g}}

\newcommand{\bl}{\mathfrak{L}}

\newcommand{\ga}{\mathbb{A}}
\newcommand{\gb}{\mathbb{B}}

\newcommand{\mR}{\mathfrak{R}}

\newcommand{\lb}{\label}

\newcommand{\vf}{\varphi}

\newcommand{\Com}{\mathbb{C}}
\newcommand{\fik}{\mathbb{K}}

\newcommand{\id}{\mathrm{Id}}

\newcommand{\un}{\mathrm{U}}

\DeclareMathOperator{\fd}{\mathbb{F}}
\DeclareMathOperator{\fds}{\mathbb{F}}

\newcommand{\zcs}{\mathcal{V}}

\newcommand{\kne}{\mathrm{KN}}
\newcommand{\cur}{\mathrm{C}}

\newcommand{\oc}{p}
\newcommand{\ocs}{p}
\newcommand{\sm}{N}
\newcommand{\hrf}{\mu}

\newcommand{\nv}{m}
\newcommand{\eo}{d}
\newcommand{\kd}{q}

\newtheorem{theorem}{Theorem}
\newtheorem{proposition}{Proposition}
\newtheorem{lemma}{Lemma}

\newtheorem{conjecture}{Conjecture}

\theoremstyle{definition}
\newtheorem{definition}{Definition}
\newtheorem{example}{Example}
\newtheorem{remark}{Remark}

\begin{document}

\title[On Lie algebras responsible for integrability]
{On Lie algebras responsible for integrability\\ 
of (1+1)-dimensional scalar evolution PDEs}
\date{}

\author{Sergei Igonin}
\address{Center of Integrable Systems, 
P.G. Demidov Yaroslavl State University, Yaroslavl, Russia, \\
\textup{E-mail address: s-igonin@yandex.ru}}

\author{Gianni Manno}
\address{Dipartimento di Scienze Matematiche, Politecnico di Torino, \\
Corso Duca degli Abruzzi 24, 10129 Torino, Italy\\
\textup{E-mail address: giovanni.manno@polito.it}}

\begin{abstract}
Zero-curvature representations (ZCRs) are one of the main tools in the theory of integrable PDEs. In particular, Lax pairs for (1+1)-dimensional PDEs can be interpreted as ZCRs.

In~\cite{scal13}, for any (1+1)-dimensional scalar evolution equation~$\CE$, we defined a family of Lie algebras~$\fds(\CE)$ which are responsible for all ZCRs of~$\CE$ in the following sense. Representations of the algebras~$\fds(\CE)$ classify all ZCRs of the equation~$\CE$ up to local gauge transformations. In~\cite{sbt18} we showed that, using these algebras, one obtains necessary conditions for existence of a B\"acklund transformation between two given equations. The algebras~$\fds(\CE)$ are defined in terms of generators and relations.

In this paper we show that, using the algebras~$\fds(\CE)$, one obtains some necessary conditions for integrability of (1+1)-dimensional scalar evolution PDEs, where integrability is understood in the sense of soliton theory. Using these conditions, we prove non-integrability for some scalar evolution PDEs of order~$5$. Also, we prove a result announced in~\cite{scal13} on the structure of the algebras~$\fds(\CE)$ for certain classes of equations of orders~$3$,~$5$,~$7$, which include KdV, mKdV, Kaup-Kupershmidt, Sawada-Kotera type equations. Among the obtained algebras for equations considered in this paper and in~\cite{sbt18}, one finds infinite-dimensional Lie algebras of certain polynomial matrix-valued functions 
on affine algebraic curves of genus~$1$ and~$0$.

In this approach, ZCRs may depend on partial derivatives of arbitrary order, which may be higher than the order of the equation~$\CE$. The algebras~$\fds(\CE)$ generalize Wahlquist-Estabrook prolongation algebras, which are responsible for a much smaller class of ZCRs.
\end{abstract}



\maketitle

\section{Introduction}
\lb{subsint1}

Zero-curvature representations (ZCRs)
belong to the main tools in the theory of integrable nonlinear partial differential equations
(see, e.g.,~\cite{zakh-shab,ft}).
In particular, Lax pairs for (1+1)-dimensional partial differential equations (PDEs) can be interpreted as ZCRs.
This paper is a sequel of~\cite{scal13} and is part of a research program
on investigating the structure of ZCRs for PDEs of various types. 
(However, the present paper can be studied independently of~\cite{scal13}.)

Here we consider (1+1)-dimensional scalar evolution equations
\beq
\label{eveq_intr}
u_t=F(x,t,u_0,u_1,\dots,u_{\eo}),\qquad\quad  
u=u(x,t),
\ee
where one uses the notation
\beq
\lb{dnot}
u_t=\frac{\pd u}{\pd t},\qquad\quad u_0=u,\qquad\quad
u_k=\frac{\pd^k u}{\pd x^k},\qquad\quad
k\in\zp.
\ee
The number $\eo\in\zsp$ in~\er{eveq_intr} is such that 
the function $F$ may depend only on $x$, $t$, $u_k$ for $k\le\eo$. 
The symbols $\zsp$ and $\zp$ denote the sets of positive and 
nonnegative integers respectively.

The methods of this paper can be applied 
also to (1+1)-dimensional multicomponent evolution PDEs, 
which are discussed in Remark~\ref{multevol}. 

\begin{remark}
\lb{rfxtu}
When we consider a function $Q=Q(x,t,u_0,u_1,\dots,u_l)$ 
for some $l\in\zp$, we always assume that this function is analytic 
on an open subset of the space $V$ with the coordinates 
$x,t,u_0,u_1,\dots,u_l$.
For example, $Q$ may be a meromorphic function,  
because a meromorphic function is analytic on some open subset.
If we say that $Q$ is defined on a neighborhood of a point $a\in V$,
we assume that the function $Q$ is analytic on this neighborhood.
\end{remark}

PDEs of the form~\er{eveq_intr} have attracted a lot of attention in the last $50$ years and 
have been a source of many remarkable results on integrability. 
In particular, some types of equations~\er{eveq_intr} possessing 
higher-order symmetries and conservation laws have been 
classified (see, e.g.,~\cite{mesh-sok13,mikh91,sand-wang2009} and references therein).
However, the problem of complete understanding of all integrability properties 
for equations~\er{eveq_intr} is still far from being solved.  

Examples of integrable PDEs of the form~\er{eveq_intr}  
include the Korteweg-de Vries (KdV), Krichever-Novikov~\cite{krich80,svin-sok83}, 
Kaup-Kupershmidt~\cite{kaup80}, Sawada-Kotera~\cite{sk74} (Caudrey-Dodd-Gibbon~\cite{cdg76}) 
equations (these equations are discussed below). 
Many more examples can be found in~\cite{mesh-sok13,mikh91,sand-wang2009} 
and references therein. 

In the present paper, integrability is understood 
in the sense of soliton theory and the inverse scattering method,
relying on the use of ZCRs.  
(This is sometimes called S-integrability.) 
As discussed in Remark~\ref{rsymcl}, 
this approach to integrability is not equivalent to the approach 
of symmetries and conservation laws.

\begin{definition}
\lb{dzcr}
Let $\mg$ be a finite-dimensional Lie algebra.  
For an equation of the form~\er{eveq_intr}, 
a \emph{zero-curvature representation \textup{(}ZCR\textup{)} 
with values in~$\mg$} is given by $\mg$-valued functions
\beq
\lb{mnoc}
A=A(x,t,u_0,u_1,\dots,u_\ocs),\qquad\quad B=B(x,t,u_0,u_1,\dots,u_{\ocs+\eo-1})
\ee
satisfying
\beq
\lb{mnzcr}
D_x(B)-D_t(A)+[A,B]=0.
\ee
The \emph{total derivative operators} $D_x$, $D_t$ in~\er{mnzcr} are 
\beq
\lb{evdxdt}
D_x=\frac{\pd}{\pd x}+\sum_{k\ge 0} u_{k+1}\frac{\pd}{\pd u_k},\qquad\qquad
D_t=\frac{\pd}{\pd t}+\sum_{k\ge 0} D_x^k\big(F(x,t,u_0,u_1,\dots,u_{\eo})\big)\frac{\pd}{\pd u_k}.
\ee
The number $\ocs$ in~\er{mnoc} is such that  
the function $A$ may depend only on the variables $x$, $t$, $u_{k}$ for $k\le\ocs$.
Then equation~\er{mnzcr} implies that 
the function $B$ may depend only on $x$, $t$, $u_{k'}$ for $k'\le\ocs+\eo-1$.

Such ZCRs are said to be \emph{of order~$\le\ocs$}. 
In other words, a ZCR given by $A$, $B$ is of order~$\le\ocs$ iff 
$\dfrac{\pd}{\pd u_l}(A)=0$ for all $l>\ocs$. 

The right-hand side $F=F(x,t,u_0,\dots,u_{\eo})$ 
of~\er{eveq_intr} appears in condition~\er{mnzcr}, 
because $F$ appears in the formula for the operator $D_t$ in~\er{evdxdt}.
Note that~\er{mnzcr} can be written as $[D_x+A,\,D_t+B]=0$, because $[D_x,D_t]=0$.
\end{definition}

\begin{remark}
\lb{rinfzcr}
The methods of this paper are applicable also to 
ZCRs with values in infinite-dimensional Lie algebras.
Such ZCRs are discussed in Section~\ref{szcrinf} and in Section~\ref{szcrsak}.
\end{remark}

In~\cite{scal13} and in this paper we study the following problem. 
How to describe all ZCRs~\er{mnoc},~\er{mnzcr} for a given equation~\er{eveq_intr}? 

In the case when $\ocs=0$ and the functions $F$, $A$, $B$ do not depend on $x$, $t$, 
a partial answer to this question is provided by the Wahlquist-Estabrook prolongation
method (WE method for short). 
Namely, for a given equation of the form $u_t=F(u_0,u_1,\dots,u_{\eo})$, 
the WE method constructs a Lie algebra so that ZCRs of the form
\beq
\lb{wecov}
A=A(u_0),\qquad B=B(u_0,u_1,\dots,u_{\eo-1}),\qquad
D_x(B)-D_t(A)+[A,B]=0
\ee
correspond to representations of this algebra (see, e.g.,~\cite{Prol,kaup_prol80,dodd,mll-2012}). 
It is called the \emph{Wahlquist-Estabrook prolongation algebra}.
Note that in~\er{wecov} the function $A=A(u_0)$ depends only on~$u_0$.

To study the general case of ZCRs~\er{mnoc},~\er{mnzcr} with arbitrary $\ocs$ 
for any equation~\er{eveq_intr}, 
we need to consider gauge transformations, which are defined below. 

Without loss of generality, one can assume that $\mg$ is a Lie subalgebra 
of $\gl_\sm$ for some $\sm\in\zsp$, 
where $\gl_\sm$ is the algebra of 
$\sm\times\sm$ matrices with entries from $\mathbb{R}$ or $\mathbb{C}$. 
So our considerations are applicable to both cases $\gl_\sm=\gl_\sm(\mathbb{R})$ 
and $\gl_\sm=\gl_\sm(\mathbb{C})$.
And we denote by $\mathrm{GL}_\sm$ the group of invertible $\sm\times\sm$ matrices.

Let $\fik$ be either $\Com$ or $\mathbb{R}$.
Then $\gl_\sm=\gl_\sm(\fik)$ and $\mathrm{GL}_\sm=\mathrm{GL}_\sm(\fik)$.
In this paper, all algebras are supposed to be over the field~$\fik$.

\begin{definition}
\lb{dlggt}
Let $\mathcal{G}\subset\mathrm{GL}_\sm$ be the connected matrix Lie group 
corresponding to the Lie algebra $\mg\subset\gl_\sm$.
(That is, $\mathcal{G}$ is the connected 
immersed Lie subgroup of $\mathrm{GL}_\sm$ 
corresponding to the Lie subalgebra $\mg\subset\gl_\sm$.)
A \emph{gauge transformation} is given by a matrix-function 
$G=G(x,t,u_0,u_1,\dots,u_l)$ with values in~$\mathcal{G}$.
Here $l$ can be any nonnegative integer.

For any ZCR~\er{mnoc},~\er{mnzcr} and 
any gauge transformation $G=G(x,t,u_0,\dots,u_l)$, the functions 
\beq
\lb{mnprint}
\tilde{A}=GAG^{-1}-D_x(G)\cdot G^{-1},\qquad\qquad
\tilde{B}=GBG^{-1}-D_t(G)\cdot G^{-1}
\ee
satisfy $D_x(\tilde{B})-D_t(\tilde{A})+[\tilde{A},\tilde{B}]=0$ and, therefore, form a ZCR.
(This is explained in Remark~\ref{remzcrgt}.)

Since $A$, $B$ take values in~$\mg$ and $G$ 
takes values in~$\mathcal{G}$, 
the functions $\tilde{A}$, $\tilde{B}$ take values in~$\mg$.

The ZCR~\er{mnprint} is said to be \emph{gauge equivalent} to the ZCR~\er{mnoc},~\er{mnzcr}. 
For a given equation~\er{eveq_intr}, formulas~\er{mnprint} determine an action of the group 
of gauge transformations on the set of ZCRs of this equation.
\end{definition}

\begin{remark}
\lb{rogt}
According to Definition~\ref{dlggt},
we study gauge transformations with values in~$\mathcal{G}$.
Alternatively, one can take some other Lie group 
$\tilde{\mathcal{G}}\subset\mathrm{GL}_\sm$ whose Lie algebra is $\mg$
and consider gauge transformations with values in~$\tilde{\mathcal{G}}$.
The results of this paper will remain valid, if one replaces 
$\mathcal{G}$ by $\tilde{\mathcal{G}}$ everywhere.

We would like to emphasize that equation~\er{eveq_intr} remains fixed 
and does not change under the action of gauge transformations.
Also, we do not have any action on solutions $u(x,t)$ of equation~\er{eveq_intr}.
In the literature on integrable PDEs, 
other authors sometimes consider transformations of different nature 
with different properties and call them gauge transformations.
So, when one speaks about gauge transformations, one should carefully 
define what they are.
\end{remark}

The WE method does not use gauge transformations in a systematic way. 
In the classification of ZCRs~\er{wecov} this is acceptable, 
because the class of ZCRs~\er{wecov} is relatively small.  

The class of ZCRs~\er{mnoc},~\er{mnzcr} is much larger than that of~\er{wecov}.
Gauge transformations play a very important role in the classification 
of ZCRs~\er{mnoc},~\er{mnzcr}. 
Because of this, the classical WE method does not produce satisfactory results 
for~\er{mnoc},~\er{mnzcr}, especially in the case~$\ocs>0$.  

To overcome this problem, 
in~\cite{scal13} we found a normal form for ZCRs~\er{mnoc},~\er{mnzcr}
with respect to the action of the group of gauge transformations.
Using the normal form of ZCRs, for any given equation~\er{eveq_intr},
in~\cite{scal13} we defined a Lie algebra $\fds^{\ocs}$ for each $\ocs\in\zp$ so  
that the following property holds. 

For every finite-dimensional Lie algebra $\mg$, 
any $\mg$-valued ZCR~\er{mnoc},~\er{mnzcr} of order~$\le\ocs$  
is locally gauge equivalent to the ZCR arising from a homomorphism 
$\fds^{\ocs}\to\mg$. 

More precisely, as discussed below, 
in~\cite{scal13} we defined a Lie algebra $\fds^{\ocs}$ for each $\ocs\in\zp$ 
and each point~$a$ of the infinite prolongation~$\CE$ of equation~\er{eveq_intr}. 
So the full notation for the algebra is $\fds^{\ocs}(\CE,a)$.  
The definition of $\fd^{\oc}(\CE,a)$ from~\cite{scal13}
is recalled in Section~\ref{btcsev} of the present paper.

The family of Lie algebras $\fds(\CE)$ mentioned in the abstract 
of this paper consists of the algebras $\fds^{\ocs}(\CE,a)$ 
for all $\ocs\in\zp$, $a\in\CE$.

Recall that the \emph{infinite prolongation} $\CE$ of equation~\er{eveq_intr} 
is an infinite-dimensional manifold with the coordinates 
$x$, $t$, $u_k$ for $k\in\zp$.
The precise definitions of the manifold $\CE$ and the algebras $\fds^{\ocs}(\CE,a)$ 
for any equation~\er{eveq_intr} are presented in Section~\ref{btcsev}. 
For every $\ocs\in\zp$ and $a\in\CE$, 
the algebra $\fds^{\ocs}(\CE,a)$ is defined in terms of generators and relations.
(To clarify the main idea, in Example~\ref{edfoc1} we consider the case $\ocs=1$.)

For every finite-dimensional Lie algebra $\mg$, 
homomorphisms $\fds^{\ocs}(\CE,a)\to\mg$ classify (up to gauge equivalence) 
all $\mg$-valued ZCRs~\er{mnoc},~\er{mnzcr} of order~$\le\ocs$, 
where functions $A$, $B$ are defined on a neighborhood of 
the point $a\in\CE$. See Section~\ref{btcsev} for more details. 

According to Section~\ref{btcsev}, 
the algebras $\fds^{\ocs}(\CE,a)$ for $\ocs\in\zp$ 
are arranged in a sequence of surjective homomorphisms 
\beq
\lb{intfdoc1}
\dots\to\fds^{\ocs}(\CE,a)\to\fds^{\ocs-1}(\CE,a)\to\dots\to\fds^1(\CE,a)\to\fds^0(\CE,a).
\ee
According to Remark~\ref{fdrzcr}, for each $\oc\in\zsp$, 
the algebra $\fds^{\ocs}(\CE,a)$ is responsible for ZCRs of order~$\le\ocs$, 
and the algebra $\fds^{\ocs-1}(\CE,a)$ is responsible for ZCRs of order~$\le\ocs-1$.    
The surjective homomorphism $\fds^{\ocs}(\CE,a)\to\fds^{\ocs-1}(\CE,a)$ in~\er{intfdoc1}
reflects the fact that any ZCR of order~$\le\ocs-1$ is at the same time of order~$\le\ocs$.
The homomorphism $\fds^{\ocs}(\CE,a)\to\fds^{\ocs-1}(\CE,a)$ is 
defined by formulas~\er{fdhff}, using generators of the algebras 
$\fds^{\ocs}(\CE,a)$, $\fds^{\ocs-1}(\CE,a)$.

Using $\fds^{\ocs}(\CE,a)$, 
we obtain some necessary conditions for integrability of equations~\er{eveq_intr}  
and necessary conditions for existence of a B\"acklund transformation 
between two given equations. 
To get such results, one needs to study certain properties 
of ZCRs~\er{mnoc},~\er{mnzcr} with arbitrary~$\ocs$, 
and we do this by means of the algebras $\fds^{\ocs}(\CE,a)$.
As explained above, the classical WE method 
(which studies ZCRs of the form~\er{wecov}) is not sufficient for this.

Applications of $\fds^{\ocs}(\CE,a)$ to obtaining 
necessary conditions for integrability of equations~\er{eveq_intr} 
are presented in Section~\ref{secicn}. 
Examples of the use of these conditions in 
proving non-integrability for some equations of order~$5$
are given in Section~\ref{secicn} as well.
Applications of $\fds^{\ocs}(\CE,a)$ to the theory of 
B\"acklund transformations are described in~\cite{sbt18}. 
See also Remark~\ref{rncebt} below.

Furthermore, we present a number of results on the structure 
of the algebras $\fds^{\ocs}(\CE,a)$ 
for some classes of scalar evolution PDEs of orders~$3$,~$5$,~$7$ 
and concrete examples.
The KdV equation is considered in~\cite{scal13} and in Example~\ref{parkdv} below.
The Krichever-Novikov equation is discussed in Proposition~\ref{fdockn}, 
which is proved in~\cite{sbt18}.
In Section~\ref{sbhheq} we study the algebras $\fds^{\ocs}(\CE,a)$ 
and integrability properties for a parameter-dependent $5$th-order scalar evolution equation,
which was considered by A.~P.~Fordy~\cite{fordy-hh} 
in connection with the H\'enon-Heiles system. 
The problem to study this equation was suggested to us by A.~P.~Fordy.

In the theory of integrable (1+1)-dimensional PDEs, 
one is especially interested in ZCRs depending on a parameter.
That is, one studies ZCRs of the form
\beq
\notag
A=A(\la,x,t,u_0,\dots,u_\oc),\quad 
B=B(\la,x,t,u_0,\dots,u_{\oc+\eo-1}),\quad
D_x(B)-D_t(A)+[A,B]=0,
\ee
where $\mg$-valued functions $A$, $B$ depend on $x$, $t$, $u_k$ 
and a parameter $\la$.
For a given equation~\er{eveq_intr}, existence of 
a nontrivial parameter-dependent ZCR is reflected in the structure of 
the algebras $\fds^{\oc}(\CE,a)$ of equation~\er{eveq_intr}.
This is illustrated by Examples~\ref{parkdv},~\ref{parkn}.

In this paper we mostly study equations of the form 
\beq
\lb{utukd}
u_t=u_{2\kd+1}+f(x,t,u_0,u_1,\dots,u_{2\kd-1}),\qquad\qquad \kd\in\{1,2,3\},
\ee
where $f$ is an arbitrary function. 
Examples of such PDEs include 
\begin{itemize}
\item the KdV equation $u_t=u_3+u_0u_1$, 
\item the modified KdV (mKdV) equation $u_t=u_3+u_0^2u_1$,
\item the Kaup-Kupershmidt equation~\cite{kaup80} $u_t=u_5+10u_0u_3+25u_1u_2+20u_0^2u_1$, 
\item the Sawada-Kotera equation~\cite{sk74} $u_t=u_5+5u_0u_3+5u_1u_2+5u_0^2u_1$, 
which is sometimes called the Caudrey-Dodd-Gibbon equation~\cite{cdg76}.
\end{itemize}
Many more examples of integrable PDEs of this type 
can be found in~\cite{mesh-sok13,mikh91} and references therein.

\begin{remark}
A classification of equations of the form 
$$
u_t=u_3+g(x,u_0,u_1,u_2),\qquad\qquad
u_t=u_5+g(u_0,u_1,u_2,u_3,u_4)
$$
satisfying certain integrability conditions related 
to generalized symmetries and conservation laws is presented in~\cite{mesh-sok13}. 
We study the problem of describing 
all ZCRs~\er{mnoc},~\er{mnzcr} for a given equation~\er{eveq_intr}.
This problem is very different from the description of 
generalized symmetries and conservation laws.
\end{remark}

Let $\bl$, $\bl_1$, $\bl_2$ be Lie algebras. 
One says that \emph{$\bl_1$ is obtained from $\bl$ by central extension} 
if there is an ideal $\mathfrak{I}\subset\bl_1$ such that 
$\mathfrak{I}$ is contained in the center of $\bl_1$ and $\bl_1/\mathfrak{I}\cong\bl$. 
Note that $\mathfrak{I}$ may be of arbitrary dimension. 

We say that 
\emph{$\bl_2$ is obtained from $\bl$ 
by applying several times the operation of central extension} 
if there is a finite collection of Lie algebras $\mg_0,\mg_1,\dots,\mg_k$ such 
that $\mg_0\cong\bl$, $\mg_k\cong\bl_2$ and 
$\mg_i$ is obtained from $\mg_{i-1}$ by central extension for each  
$i=1,\dots,k$. 

Equations of the form~\er{utukd} are considered in Theorem~\ref{thcenter}.
Some consequences of Theorem~\ref{thcenter} are summarized in Remark~\ref{remfdceukd}.

\begin{remark}
\lb{remfdceukd}
Theorem~\ref{thcenter} implies that, 
for any equation of the form~\er{utukd} with $\kd\in\{1,2,3\}$, 
\begin{itemize}
\item for every $\ocs\ge\kd+\delta_{\kd,3}$ the algebra 
$\fds^\ocs(\CE,a)$ is obtained from $\fds^{\ocs-1}(\CE,a)$ 
by central extension, 
\item for every $\ocs\ge\kd+\delta_{\kd,3}$ the algebra $\fds^\ocs(\CE,a)$ 
is obtained from $\fds^{\kd-1+\delta_{\kd,3}}(\CE,a)$ 
by applying several times the operation of central extension. 
\end{itemize}
Here $\delta_{\kd,3}$ is the Kronecker delta. 
So $\delta_{3,3}=1$, and $\delta_{\kd,3}=0$ if $\kd\neq 3$.
\end{remark}

Theorem~\ref{thcenter} was announced without proof in~\cite{scal13}.
In Section~\ref{seckdvtype} we give a detailed proof for it.

Applications of Theorem~\ref{thcenter} to obtaining necessary conditions 
for integrability of equations~\er{utukd} are presented in Section~\ref{secicn}.
Results similar to Theorem~\ref{thcenter} can be proved for many other evolution equations 
as well. See, e.g., Proposition~\ref{fdockn} about the Krichever-Novikov equation.

Other approaches to the study of 
the action of gauge transformations on ZCRs can be found 
in~\cite{marvan93,marvan97,marvan2010,sakov95,sakov2004,sebest2008} 
and references therein.
For a given ZCR with values in a matrix Lie algebra $\mg$, 
the papers~\cite{marvan93,marvan97,sakov95} define 
certain $\mg$-valued functions, which transform by conjugation 
when the ZCR transforms by gauge. 
Applications of these functions to construction and classification of 
some types of ZCRs are described  
in~\cite{marvan93,marvan97,marvan2010,sakov95,sakov2004,sebest2008}.

To our knowledge, 
the theory of~\cite{marvan93,marvan97,marvan2010,sakov95,sakov2004,sebest2008} 
does not produce any infinite-dimensional Lie algebras responsible for ZCRs. 
So this theory does not contain the algebras $\fds^\oc(\CE,a)$.

\section{Preliminaries}
\lb{spr}

We continue to use the notations introduced in Section~\ref{subsint1}.
In particular, $\CE$ is the infinite prolongation of equation~\er{eveq_intr}.
According to Definition~\ref{dipr} in Section~\ref{btcsev}, 
$\CE$ is an infinite-dimensional manifold with the coordinates 
$x$, $t$, $u_k$ for $k\in\zp$.

We suppose that the variables $x$, $t$, $u_k$ take values in $\fik$,
where $\fik$ is either $\Com$ or $\mathbb{R}$.
A point $a\in\CE$ is determined by the values of the coordinates 
$x$, $t$, $u_k$ at $a$. Let
\begin{equation}
\notag
a=(x=x_a,\,t=t_a,\,u_k=a_k)\,\in\,\CE,\qquad\qquad x_a,\,t_a,\,a_k\in\fik,\qquad k\in\zp,
\end{equation}
be a point of $\CE$.
In other words, the constants $x_a$, $t_a$, $a_k$ are the coordinates 
of the point $a\in\CE$ in the coordinate system $x$, $t$, $u_k$.

The general theory of the Lie algebras $\fd^{\oc}(\CE,a)$, $\oc\in\zp$,
is presented in Section~\ref{btcsev}.
Before describing the general theory, we would like to discuss some examples and applications.

\begin{example}
\lb{edfoc1}
To clarify the definition of $\fd^{\oc}(\CE,a)$ presented in Section~\ref{btcsev}, 
let us consider the case $\oc=1$. 
To this end, 
we fix an equation~\er{eveq_intr} and study ZCRs of order~$\le 1$ for this equation. 

According to Theorem~\ref{thnfzcr}, any ZCR of order~$\le 1$ 
\beq
\lb{zcru1}
A=A(x,t,u_0,u_1),\qquad B=B(x,t,u_0,u_1,\dots,u_{\eo}),\qquad
D_x(B)-D_t(A)+[A,B]=0
\ee
on a neighborhood of $a\in\CE$ is gauge equivalent to a ZCR of the form 
\begin{gather}
\lb{nfzcr}
\tilde{A}=\tilde{A}(x,t,u_0,u_1),\qquad \tilde{B}=\tilde{B}(x,t,u_0,u_1,\dots,u_{\eo}),\\
\lb{nfzcreq}
D_x(\tilde{B})-D_t(\tilde{A})+[\tilde{A},\tilde{B}]=0,\\
\lb{nfab}
\frac{\pd\tilde{A}}{\pd u_1}(x,t,u_0,a_1)=0,\qquad 
\tilde{A}(x,t,a_0,a_1)=0,\qquad\tilde{B}(x_a,t,a_0,a_1,\dots,a_{\eo})=0. 
\end{gather}
Moreover, according to Theorem~\ref{thnfzcr},
for any given ZCR of the form~\er{zcru1}, 
on a neighborhood of $a\in\CE$ there is a unique 
gauge transformation $G=G(x,t,u_0,\dots,u_l)$ such that
the functions $\tilde{A}=GAG^{-1}-D_x(G)\cdot G^{-1}$,
$\tilde{B}=GBG^{-1}-D_t(G)\cdot G^{-1}$ 
satisfy \er{nfzcr}, \er{nfzcreq}, \er{nfab}  
and $G(x_a,t_a,a_0,\dots,a_l)=\mathrm{Id}$, 
where $\mathrm{Id}$ is the identity matrix.

In the case of ZCRs of order~$\le 1$, 
this gauge transformation $G$ depends on $x$, $t$, $u_0$, 
so $G=G(x,t,u_0)$. 
In a similar result about ZCRs of order~$\le\oc$, 
which is described in Theorem~\ref{thnfzcr}, the corresponding 
gauge transformation depends on $x$, $t$, $u_0,\dots,u_{\oc-1}$.

Therefore, we can say that properties~\er{nfab} determine a 
normal form for ZCRs~\er{zcru1} 
with respect to the action of the group of gauge transformations 
on a neighborhood of $a\in\CE$.

A similar normal form for ZCRs~\er{mnoc},~\er{mnzcr} 
with arbitrary $\oc$ is described in Theorem~\ref{thnfzcr}  
and Remark~\ref{rnfzcr}.

Since the functions $\tilde{A}$, $\tilde{B}$ from \er{nfzcr}, \er{nfzcreq}, \er{nfab}
are analytic on a neighborhood of $a\in\CE$, these functions
are represented as absolutely convergent power series
\begin{gather}
\label{aser1}
\tilde{A}=\sum_{l_1,l_2,i_0,i_1\ge 0} 
(x-x_a)^{l_1} (t-t_a)^{l_2}(u_0-a_0)^{i_0}(u_1-a_1)^{i_1}\cdot
\tilde{A}^{l_1,l_2}_{i_0,i_1},\\
\lb{bser1}
\tilde{B}=\sum_{l_1,l_2,j_0,\dots,j_{\eo}\ge 0} 
(x-x_a)^{l_1} (t-t_a)^{l_2}(u_0-a_0)^{j_0}\dots(u_{\eo}-a_{\eo})^{j_{\eo}}\cdot
\tilde{B}^{l_1,l_2}_{j_0\dots j_{\eo}}.
\end{gather}
Here $\tilde{A}^{l_1,l_2}_{i_0,i_1}$ and $\tilde{B}^{l_1,l_2}_{j_0\dots j_{\eo}}$ 
are elements of a Lie algebra, which we do not specify yet. 

Using formulas~\er{aser1},~\er{bser1}, we see that properties~\er{nfab} are equivalent to 
\beq
\lb{ab000int}
\tilde{A}^{l_1,l_2}_{i_0,1}=
\tilde{A}^{l_1,l_2}_{0,0}=
\tilde{B}^{0,l_2}_{0\dots 0}=0
\qquad\qquad\forall\,l_1,l_2,i_0\in\zp.
\ee
To define $\fd^1(\CE,a)$ in terms of generators and relations, 
we regard $\tilde{A}^{l_1,l_2}_{i_0,i_1}$, 
$\tilde{B}^{l_1,l_2}_{j_0\dots j_{\eo}}$ from~\er{aser1},~\er{bser1} 
as abstract symbols.  
By definition, the Lie algebra $\fd^1(\CE,a)$ is generated by the symbols 
$\tilde{A}^{l_1,l_2}_{i_0,i_1}$, $\tilde{B}^{l_1,l_2}_{j_0\dots j_{\eo}}$
for $l_1,l_2,i_0,i_1,j_0,\dots,j_{\eo}\in\zp$ so that
relations for these generators are provided by equations~\er{nfzcreq},~\er{ab000int}. 

That is, in order to get relations for the generators 
$\tilde{A}^{l_1,l_2}_{i_0,i_1}$, $\tilde{B}^{l_1,l_2}_{j_0\dots j_{\eo}}$ of the algebra $\fd^1(\CE,a)$,
we substitute~\er{aser1},~\er{bser1} in~\er{nfzcreq}, taking into account~\er{ab000int}.
A more detailed description of this construction 
is given in Section~\ref{btcsev} and in~\cite{scal13}
(with a slightly different notation for the generators).
\end{example}

\begin{example}
\lb{parkdv}
It is well known that the KdV equation $u_t=u_3+u_0u_1$ possesses 
an $\msl_2(\fik)$-valued ZCR depending polynomially on a parameter $\la$.
This is reflected in the structure of the algebras $\fds^{\oc}(\CE,a)$ 
for KdV as follows.

Consider the infinite-dimensional Lie algebra 
$\msl_2(\fik[\la])\cong \msl_2(\fik)\otimes_{\fik}\fik[\lambda]$, 
where $\fik[\lambda]$ is the algebra of polynomials in $\la$.
(If we regard $\fik$ as a rational algebraic curve with coordinate~$\la$, 
the elements of~$\msl_2(\fik[\la])$ can be identified with polynomial  
$\msl_2(\fik)$-valued functions on this rational curve.)

According to~\cite{scal13}, the algebra $\msl_2(\fik[\la])$
plays the main role in the description of $\fds^\ocs(\CE,a)$ for the KdV equation. 
Namely, it is shown in~\cite{scal13} that, for KdV, the algebras 
$\fds^\ocs(\CE,a)$ are obtained from $\msl_2(\fik[\la])$
by applying several times the operation of central extension.
In particular, $\fds^0(\CE,a)$ is isomorphic to the direct sum 
of $\msl_2(\fik[\la])$ and a $3$-dimensional abelian Lie algebra.
(In the computation of~$\fds^0(\CE,a)$ in~\cite{scal13} we use 
the fact that the structure of the Wahlquist-Estabrook prolongation algebra 
for KdV is known and contains $\msl_2(\fik[\la])$~\cite{kdv,kdv1}.)

Also, one can prove similar results on the structure of $\fds^{\oc}(\CE,a)$
for many other evolution equations possessing parameter-dependent ZCRs.
\end{example}
\begin{example}
\lb{parkn}
For any constants $e_1,e_2,e_3\in\Com$, 
one has the Krichever-Novikov equation~\cite{krich80,svin-sok83} 
\begin{equation}
\label{knedef}
  \kne(e_1,e_2,e_3)=\left\{
 u_t=u_3-\frac32\frac{(u_2)^2}{u_1}+
\frac{(u_0-e_1)(u_0-e_2)(u_0-e_3)}{u_1}\right\}.
\end{equation}
We denote by $\mathfrak{so}_3(\Com)$
the $3$-dimensional orthogonal Lie algebra over $\Com$.
According to~\cite{krich80,novikov99}, 
if $e_1\neq e_2\neq e_3\neq e_1$ then the Krichever-Novikov equation~\er{knedef}
has an $\mathfrak{so}_3(\Com)$-valued ZCR with elliptic parameter. 
One can see this in the structure of the algebras $\fds^{\oc}(\CE,a)$ as follows.

Suppose that $e_1\neq e_2\neq e_3\neq e_1$.
According to Proposition~\ref{fdockn}, which is proved in~\cite{sbt18}, 
in the description of $\fds^\ocs(\CE,a)$ for the Krichever-Novikov equation~\er{knedef}
we see an infinite-dimensional Lie algebra $\mR_{e_1,e_2,e_3}$, 
which consists of certain 
$\mathfrak{so}_3(\Com)$-valued functions on an elliptic curve.
The curve and the algebra $\mR_{e_1,e_2,e_3}$ are defined in Remark~\ref{skne}.
As discussed in Remark~\ref{skne}, the curve and the algebra $\mR_{e_1,e_2,e_3}$
were studied previously by other authors in a different context.
\end{example}

\begin{remark}
\lb{rncebt}
In~\cite{sbt18} we show that the algebras $\fds^\ocs(\CE,a)$ help to obtain 
necessary conditions for existence of a B\"acklund transformation (BT) 
between two given evolution equations. This allows us to prove 
a number of non-existence results for BTs. 
For instance, a result of this kind is presented in Proposition~\ref{scknprop}, 
which is proved in~\cite{sbt18}.

For any $e_1,e_2,e_3\in\Com$, 
we have the Krichever-Novikov equation $\kne(e_1,e_2,e_3)$ given by~\er{knedef}. 
Consider also the algebraic curve 
\beq
\lb{curez}
\cur(e_1,e_2,e_3)=\Big\{(z,y)\in\Com^2\ \Big|\ 
y^2=(z-e_1)(z-e_2)(z-e_3)\Big\}.
\ee
\begin{proposition}[\cite{sbt18}]
\label{scknprop}
Let $e_1,e_2,e_3,e'_1,e'_2,e'_3\in\Com$ such that $e_1\neq e_2\neq e_3\neq e_1$ 
and $e'_1\neq e'_2\neq e'_3\neq e'_1$.

If the curve $\cur(e_1,e_2,e_3)$ is not birationally equivalent to 
the curve $\cur(e'_1,e'_2,e'_3)$, 
then the equation $\kne(e_1,e_2,e_3)$ is not connected 
with the equation $\kne(e'_1,e'_2,e'_3)$ by any B\"acklund transformation. 

Also, if $e_1\neq e_2\neq e_3\neq e_1$, then $\kne(e_1,e_2,e_3)$ is not connected with the KdV equation by any BT. 
\end{proposition}

BTs of Miura type (differential substitutions) for~\eqref{knedef} were studied 
in~\cite{mesh-sok13,svin-sok83}. 
According to~\cite{mesh-sok13,svin-sok83}, the equation $\kne(e_1,e_2,e_3)$
is connected with the KdV equation by a BT of Miura type 
iff $e_i=e_j$ for some $i\neq j$.

Proposition~\ref{scknprop} considers the most general class of BTs, 
which is much larger than the class of 
BTs of Miura type studied in~\cite{mesh-sok13,svin-sok83}. 
The definition of BTs is given in~\cite{sbt18}, 
using a geometric approach from~\cite{nonl89}.

If $e_1\neq e_2\neq e_3\neq e_1$ and $e'_1\neq e'_2\neq e'_3\neq e'_1$, 
the curves $\cur(e_1,e_2,e_3)$ and $\cur(e'_1,e'_2,e'_3)$ are elliptic.
The theory of elliptic curves allows one to determine when 
$\cur(e_1,e_2,e_3)$ is not birationally equivalent to $\cur(e'_1,e'_2,e'_3)$.
One gets a certain algebraic condition on the numbers 
$e_1$, $e_2$, $e_3$, $e'_1$, $e'_2$, $e'_3$, 
which allows us to formulate the result of Proposition~\ref{scknprop}
more explicitly. See~\cite{sbt18} for details. 
\end{remark}

\begin{remark}
\lb{multevol}
It is possible to introduce an analog of $\fds^\ocs(\CE,a)$ for 
multicomponent evolution PDEs
\beq
\lb{uitfix}
\frac{\pd u^i}{\pd t}
=F^i(x,t,u^1,\dots,u^\nv,\,u^1_1,\dots,u^\nv_1,\dots,u^1_{\eo},\dots,u^\nv_{\eo}),
\quad 
u^i=u^i(x,t),\quad u^i_k=\frac{\pd^k u^i}{\pd x^k},\quad   
i=1,\dots,\nv. 
\ee
In this paper we study only the scalar case $\nv=1$.
The case $\nv>1$ requires much more computations, which will be presented elsewhere.
Some results for $\nv>1$ 
(including a normal form for ZCRs with respect to the action of gauge transformations
and computations of~$\fds^\ocs(\CE,a)$ for a number of multicomponent PDEs of 
Landau-Lifshitz and nonlinear Schr\"odinger types)
are sketched in the preprints~\cite{hjpa,zcrm17}.
\end{remark}

\begin{remark}
\lb{remzcrgt}
It is well known that equation~\er{mnzcr} implies 
$D_x(\tilde{B})-D_t(\tilde{A})+[\tilde{A},\tilde{B}]=0$ 
for $\tilde{A}$, $\tilde{B}$ given by~\er{mnprint}.
Indeed, formulas~\er{mnprint} yield 
$D_x+\tilde{A}=G(D_x+A)G^{-1}$ and $D_x+\tilde{B}=G(D_t+B)G^{-1}$. 
Therefore, 
\begin{multline*}
D_x(\tilde{B})-D_t(\tilde{A})+[\tilde{A},\tilde{B}]=
[D_x+\tilde{A},D_t+\tilde{B}]=[G(D_x+A)G^{-1},G(D_t+B)G^{-1}]=\\
=G[D_x+A,D_t+B]G^{-1}=G(D_x(B)-D_t(A)+[A,B])G^{-1}.
\end{multline*}
Hence the equation $D_x(B)-D_t(A)+[A,B]=0$ implies 
$D_x(\tilde{B})-D_t(\tilde{A})+[\tilde{A},\tilde{B}]=0$.
\end{remark}

\begin{remark}
\lb{skne}
In this remark we assume $\fik=\Com$. 
For any constants $e_1,e_2,e_3\in\Com$, consider the Krichever-Novikov equation~\er{knedef}.
To study the algebras $\fds^\ocs(\CE,a)$ for this equation, we need some auxiliary constructions.

Let $\Com[v_1,v_2,v_3]$ be the algebra of 
polynomials in the variables $v_1$, $v_2$, $v_3$.
Let $e_1,e_2,e_3\in\Com$ such that $e_1\neq e_2\neq e_3\neq e_1$.
Consider the ideal $\mathcal{I}_{e_1,e_2,e_3}\subset\Com[v_1,v_2,v_3]$ 
generated by the polynomials
\begin{equation}
  \label{elc}
  v_i^2-v_j^2+e_i-e_j,\qquad\qquad i,j=1,2,3.
\end{equation}

Set $E_{e_1,e_2,e_3}=\Com[v_1,v_2,v_3]/\mathcal{I}_{e_1,e_2,e_3}$. 
In other words, $E_{e_1,e_2,e_3}$ 
is the commutative associative algebra of polynomial 
functions on the algebraic curve 
in $\Com^3$ defined by the polynomials~\eqref{elc}.
(This curve is given by the equations 
$v_i^2-v_j^2+e_i-e_j=0$, $i,j=1,2,3$, 
in the space $\Com^3$ with coordinates $v_1$, $v_2$, $v_3$.)

Since we assume $e_1\neq e_2\neq e_3\neq e_1$, 
this curve is nonsingular, irreducible and is of genus~$1$, 
so this is an elliptic curve. 
It is known that the Landau-Lifshitz equation 
and the Krichever-Novikov equation
possess $\mathfrak{so}_3(\Com)$-valued ZCRs parametrized by points of this 
curve~\cite{sklyanin,ft,novikov99,ll}.
(For the Krichever-Novikov equation, the paper~\cite{novikov99} presents a ZCR 
with values in the Lie algebra $\msl_2(\Com)\cong\mathfrak{so}_3(\Com)$.)

We have the natural surjective homomorphism 
$\rho\cl\Com[v_1,v_2,v_3]\to
\Com[v_1,v_2,v_3]/\mathcal{I}_{e_1,e_2,e_3}=E_{e_1,e_2,e_3}$.
Set $\hat v_i=\rho(v_i)\in E_{e_1,e_2,e_3}$ for $i=1,2,3$.

Consider also a basis $\al_1$, $\al_2$, $\al_3$ of the Lie algebra
$\mathfrak{so}_3(\Com)$ such that 
$[\al_1,\al_2]=\al_3$, $[\al_2,\al_3]=\al_1$, $[\al_3,\al_1]=\al_2$.

Denote by $\mR_{e_1,e_2,e_3}$ the Lie subalgebra of 
$\mathfrak{so}_3(\Com)\otimes_\Com E_{e_1,e_2,e_3}$ generated by the elements
$\al_i\otimes\hat v_i$, $i=1,2,3$.
Since $\mR_{e_1,e_2,e_3}\subset\mathfrak{so}_3(\Com)\otimes_\Com E_{e_1,e_2,e_3}$, 
we can view elements of~$\mR_{e_1,e_2,e_3}$ 
as $\mathfrak{so}_3(\Com)$-valued functions on the elliptic curve in~$\Com^3$ 
determined by the polynomials~\eqref{elc}. 

Set $z=\hat v_1^2+e_1$.
As $\hat v_1^2+e_1=\hat v_2^2+e_2=\hat v_3^2+e_3$ in $E_{e_1,e_2,e_3}$, we have
$z=\hat v_1^2+e_1=\hat v_2^2+e_2=\hat v_3^2+e_3$.
It is easily seen (and is shown in~\cite{ll}) 
that the following elements form a basis for $\mR_{e_1,e_2,e_3}$
\beq
\label{rbas}
\al_i\otimes\hat v_i z^l,\quad\al_i\otimes\hat v_j\hat v_k z^l,\qquad 
i,j,k\in\{1,2,3\},\quad j<k,\quad j\neq i\neq k,\quad l\in\zp.
\ee
Since the basis~\er{rbas} is infinite,
the Lie algebra $\mR_{e_1,e_2,e_3}$ is infinite-dimensional.
It is known that the standard ZCR with elliptic parameter 
for the (fully anisotropic) Landau-Lifshitz equation can be interpreted as a ZCR 
with values in this algebra~\cite{sklyanin,holod84,ll}.

It is shown in~\cite{ll} that the Wahlquist-Estabrook prolongation algebra 
of the (fully anisotropic) Landau-Lifshitz equation is isomorphic to the 
direct sum of $\mR_{e_1,e_2,e_3}$ and a $2$-dimensional abelian Lie algebra.

According to Proposition~\ref{fdockn} below, the algebra $\mR_{e_1,e_2,e_3}$ 
shows up also in the structure of~$\fds^\ocs(\CE,a)$ for the Krichever-Novikov equation.
A proof of Proposition~\ref{fdockn} is given in~\cite{sbt18}.
This proof uses some results of~\cite{igon-martini,novikov99,ll}.
\begin{proposition}[\cite{sbt18}]
\lb{fdockn}
For any $e_1,e_2,e_3\in\Com$, 
consider the Krichever-Novikov equation $\kne(e_1,e_2,e_3)$ given by~\er{knedef}.
Let $\CE$ be the infinite prolongation of this equation. Let $a\in\CE$.
Then 
\begin{itemize}
 \item the algebra $\fds^0(\CE,a)$ is zero,
\item for each $\ocs\ge 2$, 
the algebra $\fds^\ocs(\CE,a)$ is obtained from $\fds^{\ocs-1}(\CE,a)$ 
by central extension, and the kernel of the surjective homomorphism 
$\fds^{\ocs}(\CE,a)\to\fds^{1}(\CE,a)$ from~\er{fdnn-1} is nilpotent,
\item if $e_1\neq e_2\neq e_3\neq e_1$, then $\fds^1(\CE,a)\cong\mR_{e_1,e_2,e_3}$  
and for each $\ocs\ge 2$ the algebra $\fds^{\ocs}(\CE,a)$ 
is obtained from $\mR_{e_1,e_2,e_3}$ by applying several times 
the operation of central extension.
\end{itemize}
\end{proposition}
\end{remark}

\section{ZCRs, gauge transformations, and the algebras $\fds^\oc(\CE,a)$}
\lb{btcsev}

In this section we recall some notions and results from~\cite{scal13}, adding some clarifications.

As said in Section~\ref{spr}, 
we suppose that $x$, $t$, $u_k$ take values in $\fik$, 
where $\fik$ is either $\Com$ or $\mathbb{R}$.
Let $\fik^\infty$ be the infinite-dimensional space  
with the coordinates $x$, $t$, $u_k$ for $k\in\zp$. 
The topology on $\fik^\infty$ is defined as follows. 

For each $l\in\zp$, consider the space $\fik^{l+3}$ 
with the coordinates $x$, $t$, $u_k$ for $k\le l$. 
One has the natural projection $\pi_l\cl\fik^\infty\to\fik^{l+3}$ that ``forgets'' 
the coordinates $u_{k'}$ for $k'>l$. 

Since $\fik^{l+3}$ is a finite-dimensional vector space, 
we have the standard topology on~$\fik^{l+3}$. 
For any $l\in\zp$ and any open subset $V\subset\fik^{l+3}$, 
the subset~$\pi_l^{-1}(V)\subset\fik^\infty$ 
is, by definition, open in $\fik^\infty$. 
Such subsets form a base of the topology on~$\fik^\infty$. 
In other words, we consider the smallest topology on~$\fik^\infty$ such that 
the maps $\pi_l$, $l\in\zp$, are continuous. 

For a connected open subset $W\subset\fik^\infty$, 
a function $f\cl W\to\fik$ is said to be \emph{analytic} if $f$ depends 
analytically on a finite number of the coordinates $x$, $t$, $u_k$, where $k\in\zp$.
(That is, $f$ is an analytic function of the form $f=f(x,t,u_0,\dots,u_m)$ for some $m\in\zp$.)
For an arbitrary open subset $\tilde{W}\subset\fik^\infty$, 
a function $g\cl\tilde{W}\to\fik$ is called analytic 
if $g$ is analytic on each connected component of $\tilde{W}$.

Since we have the topology on $\fik^\infty$ and the notion of 
analytic functions on open subsets of~$\fik^\infty$,
we can say that $\fik^\infty$ is an analytic manifold.

\begin{definition}
\lb{dipr}
Let $\ost\subset\fik^{\eo+3}$ be an open subset such that the function 
$F=F(x,t,u_0,u_1,\dots,u_{\eo})$ from~\er{eveq_intr} is defined on~$\ost$. 
(For instance, if the function $F$ is meromorphic on~$\fik^{\eo+3}$ then one can 
take $\ost\subset\fik^{\eo+3}$ to be the maximal open subset such that 
$F$ is analytic on~$\ost$.)

The \emph{infinite prolongation} $\CE$ of equation~\er{eveq_intr} 
is defined as follows:
$$
\CE=\pi_{\eo}^{-1}(\ost)\subset\fik^\infty.
$$
So $\CE$ is an open subset of the space $\fik^\infty$  
with the coordinates $x$, $t$, $u_k$ for $k\in\zp$. 
The topology on~$\CE$ is induced by the embedding $\CE\subset\fik^\infty$. 

As said above, we view the space $\fik^\infty$ as an analytic manifold.
Since $\CE$ is an open subset of~$\fik^\infty$, the set $\CE$ is an analytic manifold as well.
\end{definition}

\begin{example}
For any constants $e_1,e_2,e_3\in\fik$, 
we write the Krichever-Novikov equation~\er{knedef} as follows 
\begin{gather}
\label{u1kn}
u_t=F(u_0,u_1,u_2,u_3),\\
\lb{f1kn}
F(u_0,u_1,u_2,u_3)=
u_3-\frac32\frac{(u_2)^2}{u_1}+
\frac{(u_0-e_1)(u_0-e_2)(u_0-e_3)}{u_1}.
\end{gather}
Since the right-hand side of~\er{u1kn} depends on $u_k$ for $k\le 3$,
we have here $\eo=3$.

Let $\fik^6$ be the space with the coordinates 
$x$, $t$, $u_0$, $u_1$, $u_2$, $u_3$.
According to~\er{f1kn}, 
the function $F$ is defined on the open subset $\ost\subset\fik^6$ determined 
by the condition $u_1\neq 0$. 

Recall that $\fik^\infty$ is the space with the coordinates
$x$, $t$, $u_k$ for $k\in\zp$. 
We have the map $\pi_3\cl\fik^\infty\to\fik^6$
that ``forgets'' the coordinates $u_{k'}$ for $k'>3$. 
The infinite prolongation $\CE$ of equation~\er{u1kn} is the following open subset 
of $\fik^\infty$
$$
\CE=\pi_3^{-1}(\ost)=
\big\{(x,t,u_0,u_1,u_2,\dots)\in\fik^\infty\,\big|\,
u_1\neq 0\big\}.
$$
\end{example}

Consider again an arbitrary scalar evolution equation~\er{eveq_intr}.
Let $\CE$ be the infinite prolongation of~\er{eveq_intr}.

Since $\CE$ is an open subset of the space $\fik^\infty$  
with the coordinates $x$, $t$, $u_k$ for $k\in\zp$, 
a point $a\in\CE$ is determined by the values of $x$, $t$, $u_k$ at~$a$. Let
\begin{equation}
\lb{pointevs}
a=(x=x_a,\,t=t_a,\,u_k=a_k)\,\in\,\CE,\qquad\qquad x_a,\,t_a,\,a_k\in\fik,\qquad k\in\zp,
\end{equation}
be a point of $\CE$.
The constants $x_a$, $t_a$, $a_k$ are the coordinates 
of the point $a\in\CE$ in the coordinate system $x$, $t$, $u_k$.

We continue to use the notations introduced in Section~\ref{subsint1}.
In particular, $\mg\subset\gl_\sm$ is a matrix Lie algebra, 
and $\mathcal{G}\subset\mathrm{GL}_\sm$ is the connected matrix Lie group 
corresponding to $\mg$, where $\sm\in\zsp$.

According to Definition~\ref{dlggt},
a gauge transformation is a matrix-function $G=G(x,t,u_0,u_1,\dots,u_l)$ 
with values in~$\mathcal{G}$, where $l\in\zp$.
See also Remark~\ref{rogt} about gauge transformations with values 
in other matrix Lie groups.

In this section, when we speak about ZCRs, we always 
mean ZCRs of equation~\er{eveq_intr}.
For each $i=1,2$, let 
\beq
\notag
A_i=A_i(x,t,u_0,u_1,\dots),\quad 
B_i=B_i(x,t,u_0,u_1,\dots),\quad
D_x(B_i)-D_t(A_i)+[A_i,B_i]=0
\ee
be a $\mg$-valued ZCR.
The ZCR $A_1,B_1$ is said to be \emph{gauge equivalent} 
to the ZCR $A_2,B_2$ if there is a gauge transformation $G=G(x,t,u_0,\dots,u_l)$ 
such that 
$$
A_1=GA_2G^{-1}-D_x(G)\cdot G^{-1},\qquad\qquad
B_1=GB_2G^{-1}-D_t(G)\cdot G^{-1}.
$$

Let $s\in\zp$. For a function $M=M(x,t,u_0,u_1,u_2,\dots)$, the notation 
$M\,\Big|_{u_k=a_k,\ k\ge s}$ 
means that we substitute $u_k=a_k$ for all $k\ge s$ in the function $M$. 
Also, sometimes we substitute $x=x_a$ or $t=t_a$ in such functions. 
For example, if $M=M(x,t,u_0,u_1,u_2,u_3)$, then 
$$
M\,\Big|_{x=x_a,\ u_k=a_k,\ k\ge 2}=M(x_a,t,u_0,u_1,a_2,a_3).
$$

The following result is obtained in~\cite{scal13}.

\begin{theorem}[\cite{scal13}]
\lb{thnfzcr}
Let $\sm\in\zsp$ and $\oc\in\zp$. 
Let $\mg\subset\gl_\sm$ be a matrix Lie algebra and  
$\mathcal{G}\subset\mathrm{GL}_\sm$ be 
the connected matrix Lie group corresponding to $\mg\subset\gl_\sm$. 

Consider a ZCR of order~$\le\oc$ given by
\beq
\lb{guzcr}
A=A(x,t,u_0,\dots,u_\oc),\quad 
B=B(x,t,u_0,\dots,u_{\oc+\eo-1}),\quad
D_x(B)-D_t(A)+[A,B]=0
\ee
such that the functions $A$, $B$ 
are analytic on a neighborhood of $a\in\CE$ and take values in $\mg$.

Then on a neighborhood of $a\in\CE$ 
there is a unique gauge transformation $G=G(x,t,u_0,\dots,u_l)$ 
such that $G(a)=\mathrm{Id}$ and the functions 
\beq
\lb{tatbgu}
\tilde{A}=GAG^{-1}-D_x(G)\cdot G^{-1},\qquad\qquad
\tilde{B}=GBG^{-1}-D_t(G)\cdot G^{-1}
\ee
satisfy 
\begin{gather}
\label{d=0}
\frac{\pd \tilde{A}}{\pd u_s}\,\,\bigg|_{u_k=a_k,\ k\ge s}=0\qquad\quad
\forall\,s\ge 1,\\
\lb{aukak}
\tilde{A}\,\Big|_{u_k=a_k,\ k\ge 0}=0,\\
\lb{bxx0}
\tilde{B}\,\Big|_{x=x_a,\ u_k=a_k,\ k\ge 0}=0.
\end{gather}

Furthermore, one has the following.
\begin{itemize}
\item The function $G$ depends only on $x$, $t$, $u_0,\dots,u_{\oc-1}$.
\textup{(}In particular, if $\oc=0$ then $G$ depends only on $x$, $t$.\textup{)}
\item The function $G$ is analytic on a neighborhood of $a\in\CE$.
\item The functions~\er{tatbgu} take values in $\mg$ and satisfy  
\begin{gather}
\lb{tatbth}
\tilde{A}=\tilde{A}(x,t,u_0,\dots,u_\oc),\quad\qquad 
\tilde{B}=\tilde{B}(x,t,u_0,\dots,u_{\oc+\eo-1}),\\
\lb{zcrtmn}
D_x(\tilde{B})-D_t(\tilde{A})+[\tilde{A},\tilde{B}]=0.
\end{gather}
So the functions~\er{tatbgu} form 
a $\mg$-valued ZCR of order~$\le\oc$.
\end{itemize}

Note that, according to our definition of gauge transformations, 
$G$ takes values in $\mathcal{G}$. 
The property $G(a)=\mathrm{Id}$ means that $G(x_a,t_a,a_0,\dots,a_{\oc-1})=\mathrm{Id}$.
\end{theorem}

\begin{definition}
Fix a point $a\in\CE$ given by~\er{pointevs}, 
which is determined by constants $x_a$, $t_a$, $a_k$.
A ZCR 
\beq
\lb{anzcr}
\anA=\anA(x,t,u_0,u_1,\dots),\qquad 
\anB=\anB(x,t,u_0,u_1,\dots),\qquad
D_x(\anB)-D_t(\anA)+[\anA,\anB]=0
\ee
is said to be \emph{$a$-normal} if $\anA$, $\anB$ satisfy the following equations
\begin{gather}
\label{agd=0}
\frac{\pd \anA}{\pd u_s}\,\,\bigg|_{u_k=a_k,\ k\ge s}=0\qquad\quad
\forall\,s\ge 1,\\
\lb{agaukak}
\anA\,\Big|_{u_k=a_k,\ k\ge 0}=0,\\
\lb{agbxx0}
\anB\,\Big|_{x=x_a,\ u_k=a_k,\ k\ge 0}=0.
\end{gather}
\end{definition}
\begin{remark}
\lb{rnfzcr}
For example, the ZCR $\tilde{A},\tilde{B}$ 
described in Theorem~\ref{thnfzcr} is 
$a$-normal, because $\tilde{A}$, $\tilde{B}$ obey \er{d=0}, 
\er{aukak}, \er{bxx0}.
Theorem~\ref{thnfzcr} implies that any ZCR on a neighborhood of $a\in\CE$
is gauge equivalent to an $a$-normal ZCR.
Therefore, following~\cite{scal13}, we can say that 
properties \er{agd=0}, \er{agaukak}, \er{agbxx0} 
determine a normal form for ZCRs
with respect to the action of the group of gauge transformations 
on a neighborhood of $a\in\CE$.
\end{remark}

\begin{remark}
\lb{abcoef0}
The functions $A$, $B$, $G$ considered in Theorem~\ref{thnfzcr}
are analytic on a neighborhood of $a\in\CE$.
Therefore, the $\mg$-valued functions $\tilde{A}$, $\tilde{B}$ 
given by~\er{tatbgu} are analytic as well.

Since $\tilde{A}$, $\tilde{B}$ are analytic and are of the form~\er{tatbth},
these functions are represented as absolutely convergent power series
\begin{gather}
\label{aser}
\tilde{A}=\sum_{l_1,l_2,i_0,\dots,i_\oc\ge 0} 
(x-x_a)^{l_1} (t-t_a)^{l_2}(u_0-a_0)^{i_0}\dots(u_\oc-a_\oc)^{i_\oc}\cdot
\tilde{A}^{l_1,l_2}_{i_0\dots i_\oc},\\
\lb{bser}
\tilde{B}=\sum_{l_1,l_2,j_0,\dots,j_{\oc+\eo-1}\ge 0} 
(x-x_a)^{l_1} (t-t_a)^{l_2}(u_0-a_0)^{j_0}\dots(u_{\oc+\eo-1}-a_{\oc+\eo-1})^{j_{\oc+\eo-1}}\cdot
\tilde{B}^{l_1,l_2}_{j_0\dots j_{\oc+\eo-1}},\\
\notag
\tilde{A}^{l_1,l_2}_{i_0\dots i_\oc},\,\tilde{B}^{l_1,l_2}_{j_0\dots j_{\oc+\eo-1}}\in\mg.
\end{gather}

For each $k\in\zsp$, we set 
\beq
\lb{defzcsoc}
\zcs_k=\Big\{(i_0,\dots,i_{k})\in\zp^{k+1}\ \Big|\ \exists\,r\in\{1,\dots,k\}\,\ 
\text{such that}\,\ i_r=1,\,\ i_q=0\,\ \forall\,q>r\Big\}.
\ee
In other words, for $k\in\zsp$ and $i_0,\dots,i_{k}\in\zp$, one has $(i_0,\dots,i_{k})\in\zcs_k$ iff
there is $r\in\{1,\dots,k\}$ such that 
$$
(i_0,\dots,i_{r-1},i_r,i_{r+1},\dots,i_{k})=(i_0,\dots,i_{r-1},1,0,\dots,0).
$$
Set also $\zcs_0=\varnothing$. So the set $\zcs_0$ is empty.

Using formulas~\er{aser},~\er{bser}, 
we see that properties~\er{d=0},~\er{aukak}, \er{bxx0}
are equivalent to 
\beq
\lb{ab000}
\tilde{A}^{l_1,l_2}_{0\dots 0}=\tilde{B}^{0,l_2}_{0\dots 0}=0,\qquad
\tilde{A}^{l_1,l_2}_{i_0\dots i_{\oc}}=0,
\qquad (i_0,\dots,i_{\oc})\in\zcs_\oc,\qquad l_1,l_2\in\zp.
\ee
\end{remark}

\begin{remark}
\lb{psdxdtlie}
Let $\bl$ be a Lie algebra and $m\in\zp$. 
Consider a formal power series of the form 
$$
C=\sum_{l_1,l_2,i_0,\dots,i_m\ge 0} 
(x-x_a)^{l_1} (t-t_a)^{l_2}(u_0-a_0)^{i_0}\dots(u_m-a_m)^{i_m}\cdot 
C^{l_1,l_2}_{i_0\dots i_m},\qquad\quad
C^{l_1,l_2}_{i_0\dots i_m}\in\bl.
$$
Set 
\begin{gather}
\lb{dxck}
D_x(C)=\sum_{l_1,l_2,i_0,\dots,i_m} 
D_x\big((x-x_a)^{l_1} (t-t_a)^{l_2}(u_0-a_0)^{i_0}\dots(u_m-a_m)^{i_m}\big)\cdot 
C^{l_1,l_2}_{i_0\dots i_m},\\
\lb{dtck}
D_t(C)=\sum_{l_1,l_2,i_0,\dots,i_m} 
D_t\big((x-x_a)^{l_1} (t-t_a)^{l_2}(u_0-a_0)^{i_0}\dots(u_m-a_m)^{i_m}\big)\cdot 
C^{l_1,l_2}_{i_0\dots i_m}.
\end{gather}
The expressions 
\beq
\lb{dxtua}
\begin{aligned}
D_x\big((x-x_a)^{l_1} (t-t_a)^{l_2}(u_0-a_0)^{i_0}\dots(u_m-a_m)^{i_m}\big),\\
D_t\big((x-x_a)^{l_1} (t-t_a)^{l_2}(u_0-a_0)^{i_0}\dots(u_m-a_m)^{i_m}\big)
\end{aligned}
\ee
are functions of the variables $x$, $t$, $u_k$. 
Taking the corresponding Taylor series at the point~\eqref{pointevs}, 
we view~\er{dxtua} as power series. 
Then~\er{dxck},~\er{dtck} become formal power series with coefficients in $\bl$. 

According to~\er{evdxdt}, one has  
$D_t=\frac{\pd}{\pd t}+\sum_{k\ge 0} D_x^k(F)\frac{\pd}{\pd u_k}$, 
where $F=F(x,t,u_0,\dots,u_{\eo})$ is given in~\er{eveq_intr}. 
When we apply $D_t$ in~\er{dtck}, 
we view $F$ as a power series, using the Taylor series of the function~$F$.

Let $n\in\zp$ and consider another formal power series
$$
R=\sum_{q_1,q_2,j_0,\dots,j_n\ge 0} 
(x-x_a)^{q_1} (t-t_a)^{q_2}(u_0-a_0)^{j_0}\dots(u_n-a_n)^{j_n}\cdot 
R^{q_1,q_2}_{j_0\dots j_n},\qquad\quad
R^{q_1,q_2}_{j_0\dots j_n}\in\bl.
$$
Then the Lie bracket $[C,R]$ is defined in the obvious way and is a formal power series with coefficients in~$\bl$.

According to the described procedure, the expression $D_x(R)-D_t(C)+[C,R]$ is well defined and 
is a formal power series with coefficients in~$\bl$.
\end{remark} 

\begin{remark}  
The main idea of the definition of the Lie algebra $\fd^\oc(\CE,a)$  
can be informally outlined as follows. 
According to Theorem~\ref{thnfzcr} and Remark~\ref{abcoef0}, 
any ZCR~\er{guzcr} of order~$\le\oc$ is gauge equivalent 
to a ZCR given by functions $\tilde{A}$, $\tilde{B}$ 
that are of the form~\er{aser},~\er{bser} 
and satisfy~\er{zcrtmn}, \er{ab000}.

To define $\fd^\oc(\CE,a)$ in terms of generators and relations, 
one can regard $\tilde{A}^{l_1,l_2}_{i_0\dots i_\oc}$, 
$\tilde{B}^{l_1,l_2}_{j_0\dots j_{\oc+\eo-1}}$ from~\er{aser},~\er{bser} 
as abstract symbols. 
Then one can say that the Lie algebra $\fd^\oc(\CE,a)$ is generated by the symbols 
$\tilde{A}^{l_1,l_2}_{i_0\dots i_\oc}$, $\tilde{B}^{l_1,l_2}_{j_0\dots j_{\oc+\eo-1}}$
for $l_1,l_2,i_0,\dots,i_\oc,j_0,\dots,j_{\oc+\eo-1}\in\zp$ so that 
relations for these generators are provided by equations~\er{zcrtmn}, \er{ab000}.

The details of this construction are presented below.
To avoid confusion in notation, we introduce new symbols $\ga^{l_1,l_2}_{i_0\dots i_\oc}$,
$\gb^{l_1,l_2}_{j_0\dots j_{\oc+\eo-1}}$, which will be generators of the algebra $\fd^\oc(\CE,a)$.
\end{remark}

Fix $\oc\in\zp$ and consider formal power series 
\begin{gather}
\label{gasumxt}
\ga=\sum_{l_1,l_2,i_0,\dots,i_\oc\ge 0} 
(x-x_a)^{l_1} (t-t_a)^{l_2}(u_0-a_0)^{i_0}\dots(u_\oc-a_\oc)^{i_\oc}\cdot
\ga^{l_1,l_2}_{i_0\dots i_\oc},\\
\label{gbsumxt}
\gb=\sum_{l_1,l_2,j_0,\dots,j_{\oc+\eo-1}\ge 0} 
(x-x_a)^{l_1} (t-t_a)^{l_2}(u_0-a_0)^{j_0}\dots(u_{\oc+\eo-1}-a_{\oc+\eo-1})^{j_{\oc+\eo-1}}\cdot
\gb^{l_1,l_2}_{j_0\dots j_{\oc+\eo-1}},
\end{gather}
where 
\beq
\lb{gagbll}
\ga^{l_1,l_2}_{i_0\dots i_\oc},\qquad
\gb^{l_1,l_2}_{j_0\dots j_{\oc+\eo-1}},\qquad
l_1,l_2,i_0,\dots,i_\oc,j_0,\dots,j_{\oc+\eo-1}\in\zp,
\ee
are generators of a Lie algebra, which is described below.

We impose the equation
\beq
\lb{xgbtga}
D_x(\gb)-D_t(\ga)+[\ga,\gb]=0,
\ee
which is equivalent to some Lie algebraic relations for the generators~\er{gagbll}.
The left-hand side of~\er{xgbtga} is defined by the procedure described in Remark~\ref{psdxdtlie}.
Also, we impose the following condition 
\beq
\lb{gagb00}
\ga^{l_1,l_2}_{0\dots 0}=\gb^{0,l_2}_{0\dots 0}=0,\qquad 
\ga^{l_1,l_2}_{i_0\dots i_{\oc}}=0,
\qquad (i_0,\dots,i_{\oc})\in\zcs_\oc,\qquad l_1,l_2\in\zp.
\ee

\begin{definition}
\lb{dfd}
Recall that the manifold $\CE$ is the infinite prolongation
of equation~\er{eveq_intr}, and $a\in\CE$ is given by~\er{pointevs}, 
where the constants $x_a$, $t_a$, $a_k$ are the coordinates 
of the point~$a$ in the coordinate system $x$, $t$, $u_k$.
For each $\oc\in\zp$, 
the Lie algebra $\fd^\oc(\CE,a)$ is defined in terms 
of generators and relations as follows. 
The algebra $\fd^\oc(\CE,a)$ is given by the generators~\er{gagbll},
relations~\er{gagb00}, and the relations arising from~\er{xgbtga} 
when we substitute~\er{gasumxt},~\er{gbsumxt} in~\er{xgbtga}.
\end{definition}

This description of $\fd^\oc(\CE,a)$ is sufficient for the present paper.
A more detailed definition of $\fd^\oc(\CE,a)$ is given in~\cite{scal13}.
Note that condition~\er{gagb00} is equivalent to the following equations
\begin{gather}
\lb{pdgau}
\frac{\pd\ga}{\pd u_s}\,\,\bigg|_{u_k=a_k,\ k\ge s}=0\qquad\quad
\forall\,s\ge 1,\\
\lb{gaua0}
\ga\,\Big|_{u_k=a_k,\ k\ge 0}=0,\\
\lb{gbxua0}
\gb\,\Big|_{x=x_a,\ u_k=a_k,\ k\ge 0}=0.
\end{gather}

So the algebra $\fds^\ocs(\CE,a)$ is generated by the elements~\er{gagbll}.
Theorem~\ref{lemgenfdq} below, which is proved in~\cite{scal13}, says that 
the elements~\er{gal1alprop} generate the algebra $\fds^\ocs(\ce,a)$ as well.
This fact is very useful in computations of $\fds^\ocs(\ce,a)$ 
for concrete equations, because the set of the elements~\er{gal1alprop} 
is much smaller than that of~\er{gagbll}.
We will use Theorem~\ref{lemgenfdq} in Section~\ref{seckdvtype}.

\begin{theorem}[\cite{scal13}]
\lb{lemgenfdq}
The elements 
\beq
\lb{gal1alprop}
\ga^{l_1,0}_{i_0\dots i_\ocs},\qquad\qquad l_1,i_0,\dots,i_\ocs\in\zp,
\ee
generate the algebra $\fds^\ocs(\ce,a)$.
\end{theorem}

\begin{remark}
\lb{fdrzcr}
Let $\mg$ be a finite-dimensional matrix Lie algebra.
By Theorem~\ref{thnfzcr}, for any $\mg$-valued ZCR~\er{guzcr} of order~$\le\ocs$ 
on a neighborhood of $a\in\CE$, 
there is a unique gauge transformation $G$ 
such that $G(a)=\mathrm{Id}$ and the functions~\er{tatbgu} obey
\er{d=0}, \er{aukak}, \er{bxx0}.
Furthermore, Theorem~\ref{thnfzcr} says that the functions~\er{tatbgu} take values in~$\mg$ 
and satisfy \er{tatbth}, \er{zcrtmn}.

Consider the Taylor series~\er{aser},~\er{bser} of the functions~\er{tatbgu}.
Properties~\er{d=0},~\er{aukak}, \er{bxx0} are equivalent to~\er{ab000}.
Properties~\er{zcrtmn},~\er{ab000} imply that the following homomorphism 
\beq
\lb{hrfdef}
\hrf\cl\fds^\ocs(\CE,a)\to\mg,\qquad
\hrf\big(\ga^{l_1,l_2}_{i_0\dots i_\ocs}\big)=
\tilde{A}^{l_1,l_2}_{i_0\dots i_\ocs},\qquad
\hrf\big(\gb^{l_1,l_2}_{j_0\dots j_{\ocs+\eo-1}}\big)=
\tilde{B}^{l_1,l_2}_{j_0\dots j_{\ocs+\eo-1}},
\ee
is well defined.
Here 
$\tilde{A}^{l_1,l_2}_{i_0\dots i_\ocs},
\tilde{B}^{l_1,l_2}_{j_0\dots j_{\ocs+\eo-1}}\in\mg$ 
are the coefficients of the power series~\er{aser},~\er{bser}.
The definition~\er{hrfdef} of~$\hrf$ implies that
the ZCR given by~\er{aser},~\er{bser} takes
values in the Lie subalgebra $\hrf\big(\fds^\ocs(\CE,a)\big)\subset\mg$.

It is shown in~\cite{scal13} that the ZCR~\er{guzcr} is uniquely determined (up to gauge equivalence)
by the corresponding homomorphism $\hrf\cl\fds^\ocs(\CE,a)\to\mg$.

On the other hand, consider an arbitrary homomorphism $\tilde{\hrf}\cl\fds^\ocs(\CE,a)\to\mg$.
Applying $\tilde{\hrf}$ to the coefficients of the power series \er{gasumxt}, \er{gbsumxt}, 
we get the following power series with coefficients in~$\mg$
\begin{gather}
\label{hraser}
\anA=\sum_{l_1,l_2,i_0,\dots,i_\ocs} 
(x-x_a)^{l_1} (t-t_a)^{l_2}(u_0-a_0)^{i_0}\dots(u_\ocs-a_\ocs)^{i_\ocs}\cdot
\tilde{\hrf}\big(\ga^{l_1,l_2}_{i_0\dots i_\ocs}\big),\\
\lb{hrbser}
\anB=\sum_{l_1,l_2,j_0,\dots,j_{\ocs+\eo-1}} 
(x-x_a)^{l_1} (t-t_a)^{l_2}(u_0-a_0)^{j_0}\dots(u_{\ocs+\eo-1}-a_{\ocs+\eo-1})^{j_{\ocs+\eo-1}}\cdot
\tilde{\hrf}\big(\gb^{l_1,l_2}_{j_0\dots j_{\ocs+\eo-1}}\big).
\end{gather}
Since \er{gasumxt}, \er{gbsumxt} obey~\er{xgbtga}, 
the power series~\er{hraser},~\er{hrbser} satisfy $D_x(\anB)-D_t(\anA)+[\anA,\anB]=0$.
Using Definition~\ref{deffzcr}, 
we can say that the formal power series~\er{hraser},~\er{hrbser}
constitute a formal ZCR of order~$\le\oc$ with coefficients in~$\mg$.
If the power series~\er{hraser},~\er{hrbser} converge to analytic functions, 
then they constitute a $\mg$-valued ZCR of order~$\le\oc$.

The described correspondence between $\mg$-valued ZCRs 
and homomorphisms $\hrf\cl\fds^\ocs(\CE,a)\to\mg$
allows one to say that the algebra $\fd^{\oc}(\CE,a)$ is responsible for ZCRs of order~$\le\oc$.
\end{remark}

Suppose that $\oc\ge 1$. According to Definition~\ref{dfd}, 
to define the algebra $\fds^\ocs(\CE,a)$, we take formal power series~\er{gasumxt},~\er{gbsumxt}
and impose conditions~\er{xgbtga},~\er{gagb00}.
The Lie algebra $\fds^\ocs(\CE,a)$ is given by the generators 
$\ga^{l_1,l_2}_{i_0\dots i_\oc}$, $\gb^{l_1,l_2}_{j_0\dots j_{\oc+\eo-1}}$
and the relations arising from~\er{xgbtga},~\er{gagb00}. 
Similarly, to define the algebra $\fds^{\oc-1}(\CE,a)$, we take formal power series
\begin{gather*}
\hat\ga=\sum_{l_1,l_2,i_0,\dots,i_{\oc-1}} 
(x-x_a)^{l_1} (t-t_a)^{l_2}(u_0-a_0)^{i_0}\dots(u_{\oc-1}-a_{\oc-1})^{i_{\oc-1}}\cdot
\hat\ga^{l_1,l_2}_{i_0\dots i_{\oc-1}},\\
\hat\gb=\sum_{l_1,l_2,j_0,\dots,j_{\oc+\eo-2}} 
(x-x_a)^{l_1} (t-t_a)^{l_2}(u_0-a_0)^{j_0}\dots(u_{\oc+\eo-2}-a_{\oc+\eo-2})^{j_{\oc+\eo-2}}
\cdot\hat\gb^{l_1,l_2}_{j_0\dots j_{\oc+\eo-2}}
\end{gather*}
and impose the following conditions
\begin{gather}
\lb{hatgzcr}
D_x\big(\hat\gb\big)-D_t\big(\hat\ga\big)+\big[\hat\ga,\hat\gb\big]=0,\\
\lb{hatga0}
\hat\ga^{l_1,l_2}_{0\dots 0}=\hat\gb^{0,l_2}_{0\dots 0}=0,\qquad 
\hat\ga^{l_1,l_2}_{i_0\dots i_{\oc-1}}=0,
\qquad (i_0,\dots,i_{\oc-1})\in\zcs_{\oc-1},\qquad l_1,l_2\in\zp.
\end{gather}
The Lie algebra $\fds^{\oc-1}(\CE,a)$ is given by the generators 
$\hat\ga^{l_1,l_2}_{i_0\dots i_{\oc-1}}$, $\hat\gb^{l_1,l_2}_{j_0\dots j_{\oc+\eo-2}}$
and the relations arising from~\er{hatgzcr},~\er{hatga0}.

This implies that the map 
\beq
\lb{fdhff}
\ga^{l_1,l_2}_{i_0\dots i_{\oc-1}i_{\oc}}\,\mapsto\,
\delta_{0,i_{\oc}}\cdot\hat\ga^{l_1,l_2}_{i_0\dots i_{\oc-1}},\qquad\quad 
\gb^{l_1,l_2}_{j_0\dots j_{\oc+\eo-2}j_{\oc+\eo-1}}\,\mapsto\,
\delta_{0,j_{\oc+\eo-1}}\cdot\hat\gb^{l_1,l_2}_{j_0\dots j_{\oc+\eo-2}}  
\ee
determines a surjective homomorphism $\fds^{\oc}(\CE,a)\to\fds^{\oc-1}(\CE,a)$.
Here $\delta_{0,i_{\oc}}$ and $\delta_{0,j_{\oc+\eo-1}}$ are the Kronecker deltas.
We denote this homomorphism by $\vf_\oc\cl\fd^\oc(\CE,a)\to\fd^{\oc-1}(\CE,a)$.

According to Remark~\ref{fdrzcr}, the algebra $\fd^{\oc}(\CE,a)$ 
is responsible for ZCRs of order~$\le\oc$, 
and the algebra $\fd^{\oc-1}(\CE,a)$ is responsible for ZCRs of order~$\le\oc-1$.    
The constructed homomorphism $\vf_\oc\cl\fd^\oc(\CE,a)\to\fd^{\oc-1}(\CE,a)$ 
reflects the fact that any ZCR of order~$\le\oc-1$ is at the same time of order~$\le\oc$.
Thus we obtain the following sequence of surjective homomorphisms of Lie algebras
\beq
\lb{fdnn-1}
\dots\xrightarrow{\vf_{\oc+1}}
\fd^{\oc}(\CE,a)\xrightarrow{\vf_\oc}\fd^{\oc-1}(\CE,a)
\xrightarrow{\vf_{\oc-1}}\dots\xrightarrow{\vf_{2}}
\fd^1(\CE,a)\xrightarrow{\vf_{1}}\fd^0(\CE,a).
\ee

\section{Some results on $\fds^\ocs(\CE,a)$ for equations~\er{utukd}}
\lb{seckdvtype}

In this section we study the algebras~\er{fdnn-1} for equations 
of the form~\er{utukd},
where $f=f(x,t,u_0,\dots,u_{2\kd-1})$ is an arbitrary function 
and $\kd\in\{1,2,3\}$. 

Let $\CE$ be the infinite prolongation of equation~\er{utukd}. 
According to Definition~\ref{dipr}, $\CE$ is an open subset of the space $\fik^\infty$  
with the coordinates $x$, $t$, $u_k$ for $k\in\zp$. 
For equation~\er{utukd}, the total derivative operators~\er{evdxdt} are 
\beq
\label{kdvdxdt}
  D_x=\frac{\pd}{\pd x}+\sum_{k\ge 0} u_{k+1}\frac{\pd}{\pd u_k},\qquad\qquad
  D_t=\frac{\pd}{\pd t}+\sum_{k\ge 0} 
  D_x^k\big(u_{2\kd+1}+f(x,t,u_0,\dots,u_{2\kd-1})\big)\frac{\pd}{\pd u_k}.
\ee
Consider an arbitrary point $a\in\CE$ given by~\er{pointevs}, 
where the constants $x_a$, $t_a$, $a_k$ are the coordinates 
of~$a$ in the coordinate system $x$, $t$, $u_k$.

Let $\ocs\in\zsp$ such that $\ocs\ge\kd+\delta_{\kd,3}$, where $\delta_{\kd,3}$ is the Kronecker delta.
According to Definition~\ref{dfd}, the algebra $\fds^\ocs(\CE,a)$ can be described as follows.
Consider formal power series 
\begin{gather}
\label{kdvgaser}
\ga=\sum_{l_1,l_2,i_0,\dots,i_\ocs\ge 0} 
(x-x_a)^{l_1} (t-t_a)^{l_2}(u_0-a_0)^{i_0}\dots(u_\ocs-a_\ocs)^{i_\ocs}\cdot
\ga^{l_1,l_2}_{i_0\dots i_\ocs},\\
\lb{kdvgbser}
\gb=\sum_{l_1,l_2,j_0,\dots,j_{\ocs+2\kd}\ge 0} 
(x-x_a)^{l_1} (t-t_a)^{l_2}(u_0-a_0)^{j_0}\dots(u_{\ocs+2\kd}-a_{\ocs+2\kd})^{j_{\ocs+2\kd}}\cdot
\gb^{l_1,l_2}_{j_0\dots j_{\ocs+2\kd}}
\end{gather}
satisfying 
\begin{gather}
\label{relgakdv}
\ga^{l_1,l_2}_{i_0\dots i_\ocs}=0\quad \text{if}\quad \exists\,r\in\{1,\dots,\ocs\}\quad \text{such that}\quad 
i_r=1,\quad i_m=0\quad \forall\,m>r,\\
\lb{kdvaukak}
\ga^{l_1,l_2}_{0\dots 0}=0\qquad\qquad\forall\,l_1,l_2\in\zp,\\
\lb{kdvbxx0}
\gb^{0,l_2}_{0\dots 0}=0\qquad\qquad\forall\,l_2\in\zp.
\end{gather}
Then $\ga^{l_1,l_2}_{i_0\dots i_\ocs}$, $\gb^{l_1,l_2}_{j_0\dots j_{\ocs+2\kd}}$ 
are generators of the Lie algebra $\fds^\ocs(\CE,a)$, and the equation 
\beq
\label{gckdv}
D_x(\gb)-D_t(\ga)+[\ga,\gb]=0
\ee
provides relations for these generators (in addition to 
relations~\er{relgakdv},~\er{kdvaukak},~\er{kdvbxx0}). 

Condition~\er{relgakdv} is equivalent to
\beq
\label{kdvd=0}
\frac{\pd}{\pd u_s}(\ga)\,\,\bigg|_{u_k=a_k,\ k\ge s}=0\qquad\qquad
\forall\,s\ge 1.
\ee
Using~\er{kdvdxdt}, one can rewrite equation~\er{gckdv} as
\beq
\lb{kdvqzcr}
\frac{\pd}{\pd x}(\gb)+\sum_{k=0}^{\ocs+2\kd} u_{k+1}\frac{\pd}{\pd u_k}(\gb)+[\ga,\gb]
=\frac{\pd}{\pd t}(\ga)+\sum_{k=0}^\ocs
\Big(u_{k+2\kd+1}+D_x^k\big(f(x,t,u_0,\dots,u_{2\kd-1})\big)\Big)\frac{\pd}{\pd u_k}(\ga).
\ee
Here we view $f(x,t,u_0,\dots,u_{2\kd-1})$ as a power series, 
using the Taylor series of the function~$f$ at the point~\er{pointevs}.
Differentiating~\eqref{kdvqzcr} with respect to $u_{\ocs+2\kd+1}$, we obtain 
\beq
\lb{uq2buqa}
\frac{\pd}{\pd u_{\ocs+2\kd}}(\gb)=\frac{\pd}{\pd u_\ocs}(\ga).
\ee
Since $\kd\in\{1,2,3\}$, from~\er{uq2buqa} it follows that $\gb$ is of the form
\begin{equation}
\label{buq2a}
\gb=u_{\ocs+2\kd}\frac{\pd}{\pd u_\ocs}(\ga)+\gb_0(x,t,u_0,\dots,u_{\ocs+2\kd-1}),
\end{equation}
where $\gb_0(x,t,u_0,\dots,u_{\ocs+2\kd-1})$ is a power series in the variables 
$$
x-x_a,\qquad t-t_a,\qquad u_0-a_0,\qquad \dots,\qquad u_{\ocs+2\kd-1}-a_{\ocs+2\kd-1}.
$$ 

Differentiating~\eqref{kdvqzcr} with respect to $u_{\ocs+2\kd}$, $u_{\ocs+i}$ 
for $i=1,\dots,2\kd-1$ and using~\er{buq2a}, one gets 
\begin{equation*}
\frac{\pd^2}{\pd u_{\ocs}\pd u_{\ocs}}(\ga)+\frac{\pd^2}{\pd u_{\ocs+1}\pd u_{\ocs+2\kd-1}}(\gb_0)=0,\qquad
\frac{\pd^2}{\pd u_{\ocs+s}\pd u_{\ocs+2\kd-1}}(\gb_0)=0,\qquad 2\le s\le 2\kd-1.
\end{equation*}
Therefore, $\gb_0=\gb_0(x,t,u_0,\dots,u_{\ocs+2\kd-1})$ is of the form 
\beq
\label{buq2b}
\gb_0=u_{\ocs+1}u_{\ocs+2\kd-1}\Big(\frac12\delta_{\kd,1}-1\Big)
\frac{\pd^2}{\pd u_{\ocs}\pd u_{\ocs}}(\ga)
+u_{\ocs+2\kd-1}\gb_{01}(x,t,u_0,\dots,u_{\ocs})+\gb_{00}(x,t,u_0,\dots,u_{\ocs+2\kd-2}).
\ee
Here 
$\gb_{01}(x,t,u_0,\dots,u_{\ocs})$ is a power series 
in the variables 
$$
x-x_a,\quad t-t_a,\quad u_0-a_0,\quad\dots,\quad u_{\ocs}-a_{\ocs},
$$ 
and $\gb_{00}(x,t,u_0,\dots,u_{\ocs+2\kd-2})$ is a power series 
in the variables 
$$
x-x_a,\quad t-t_a,\quad 
u_0-a_0,\quad\dots,\quad u_{\ocs+2\kd-2}-a_{\ocs+2\kd-2}.
$$

\begin{lemma}
Recall that $\kd\in\{1,2,3\}$ and $\ocs\ge\kd+\delta_{\kd,3}$. We have  
\beq
\lb{dxpd2aun}
D_x\Big(\frac{\pd^2}{\pd u_{\ocs}\pd u_{\ocs}}(\ga)\Big)+
\Big[\ga,\,\frac{\pd^2}{\pd u_{\ocs}\pd u_{\ocs}}(\ga)\Big]=0.
\ee
\end{lemma}
\begin{proof}
Since $D_x=\frac{\pd}{\pd x}+\sum_{k\ge 0} u_{k+1}\frac{\pd}{\pd u_k}$, one has
\beq
\lb{undxq}
\frac{\pd}{\pd u_n}\big(D_x(Q)\big)=D_x\Big(\frac{\pd}{\pd u_n}(Q)\Big)
+\frac{\pd}{\pd u_{n-1}}(Q)\qquad\quad\forall\,n\in\zsp
\ee
for any $Q=Q(x,t,u_0,u_1,\dots,u_l)$. Here $Q$ is either a function or a power series.

In what follows we sometimes use the notation 
\beq
\notag
\gb_{u_{n}}=\frac{\pd}{\pd u_{n}}(\gb),\qquad
\gb_{u_{m}u_{n}}=\frac{\pd^2}{\pd u_{m}\pd u_{n}}(\gb),
\quad\qquad m,n\in\zp.
\ee
Using~\er{undxq}, one gets
\beq
\lb{umnb}
\frac{\pd^2}{\pd u_{m}\pd u_{n}}\big(D_x(\gb)\big)=
\frac{\pd}{\pd u_m}\Big(D_x\Big(\frac{\pd}{\pd u_n}(\gb)\Big)+
\frac{\pd}{\pd u_{n-1}}(\gb)\Big)
=D_x(\gb_{u_{m}u_{n}})+
\gb_{u_{m-1}u_{n}}+\gb_{u_{m}u_{n-1}}\quad\,
\forall\,m,n\in\zsp.
\ee
We will need also the formula
\beq
\lb{dtga}
D_t(\ga)=\frac{\pd}{\pd t}(\ga)+\sum_{k=0}^\ocs
\Big(u_{k+2\kd+1}+D_x^k\big(f(x,t,u_0,\dots,u_{2\kd-1})\big)\Big)\frac{\pd}{\pd u_k}(\ga),
\ee
which follows from~\er{kdvdxdt}.

Consider the case $\kd=1$. Then, by our assumption, $\ocs\ge 1$.
Equation~\er{utukd} reads $u_t=u_{3}+f(x,t,u_0,u_1)$. 
According to~\er{buq2a},~\er{buq2b}, for $\kd=1$ one has 
\beq
\lb{gbq1}
\gb=u_{\ocs+2}\frac{\pd}{\pd u_\ocs}(\ga)
-\frac12(u_{\ocs+1})^2\frac{\pd^2}{\pd u_{\ocs}\pd u_{\ocs}}(\ga)
+u_{\ocs+1}\gb_{01}(x,t,u_0,\dots,u_{\ocs})+\gb_{00}(x,t,u_0,\dots,u_{\ocs}).
\ee
Since we assume $\kd=1$ and $\ocs\ge 1$, formula~\er{dtga} implies 
\beq
\lb{pddtaq1}
\frac{\pd^2}{\pd u_{\ocs}\pd u_{\ocs+2}}\big(D_t(\ga)\big)=
\frac{\pd^2}{\pd u_{\ocs}\pd u_{\ocs-1}}(\ga),\qquad\quad
\frac{\pd^2}{\pd u_{\ocs+1}\pd u_{\ocs+1}}\big(D_t(\ga)\big)=0.
\ee
Using~\er{umnb},~\er{gbq1},~\er{pddtaq1}, one obtains  
\begin{multline}
\lb{pdpdpdzcr}
\Big(\frac{\pd^2}{\pd u_{\ocs}\pd u_{\ocs+2}}
-\frac12\frac{\pd^2}{\pd u_{\ocs+1}\pd u_{\ocs+1}}\Big)
\Big(D_x(\gb)-D_t(\ga)+[\ga,\gb]\Big)=\\
=D_x\Big(\gb_{u_{\ocs}u_{\ocs+2}}-\frac12\gb_{u_{\ocs+1}u_{\ocs+1}}\Big)
+\gb_{u_{\ocs-1}u_{\ocs+2}}-
\frac{\pd^2}{\pd u_{\ocs}\pd u_{\ocs-1}}(\ga)+
\Big[\ga,\,\gb_{u_{\ocs}u_{\ocs+2}}-\frac12\gb_{u_{\ocs+1}u_{\ocs+1}}\Big]
=\\
=\frac32\bigg(D_x\Big(\frac{\pd^2}{\pd u_{\ocs}\pd u_{\ocs}}(\ga)\Big)+
\Big[\ga,\,\frac{\pd^2}{\pd u_{\ocs}\pd u_{\ocs}}(\ga)\Big]\bigg).
\end{multline}
Since $D_x(\gb)-D_t(\ga)+[\ga,\gb]=0$ by~\er{gckdv}, 
equation~\er{pdpdpdzcr} 
implies~\er{dxpd2aun} in the case $\kd=1$. 

Now let $\kd=2$. Then $\ocs\ge 2$. 
Equation~\er{utukd} reads 
$u_t=u_{5}+f(x,t,u_0,u_1,u_2,u_3)$.
Using~\er{buq2a},~\er{buq2b}, for $\kd=2$ one obtains 
\beq
\lb{gbq2}
\gb=u_{\ocs+4}\frac{\pd}{\pd u_\ocs}(\ga)
-u_{\ocs+1}u_{\ocs+3}\frac{\pd^2}{\pd u_{\ocs}\pd u_{\ocs}}(\ga)
+u_{\ocs+3}\gb_{01}(x,t,u_0,\dots,u_{\ocs})+\gb_{00}(x,t,u_0,\dots,u_{\ocs+2}).
\ee
Applying the 
operator~$\dfrac{\pd^2}{\pd u_{\ocs+2}\pd u_{\ocs+3}}$ 
to equation~\er{kdvqzcr} and using~\er{gbq2}, we get 
\beq
\lb{gagbpd2}
\frac{\pd^2}{\pd u_{\ocs+2}\pd u_{\ocs+2}}(\gb_{00})-
\frac{\pd^2}{\pd u_{\ocs}\pd u_{\ocs}}(\ga)=0.
\ee
Since we assume $\kd=2$ and $\ocs\ge 2$, formula~\er{dtga} implies
\begin{gather}
\lb{ddtaq21}
\frac{\pd^2}{\pd u_{\ocs}\pd u_{\ocs+4}}\big(D_t(\ga)\big)=
\frac{\pd^2}{\pd u_{\ocs}\pd u_{\ocs-1}}(\ga),\\
\lb{ddtaq22}
\Big(-\frac{\pd^2}{\pd u_{\ocs+1}\pd u_{\ocs+3}}
+\frac12\frac{\pd^2}{\pd u_{\ocs+2}\pd u_{\ocs+2}}\Big)\big(D_t(\ga)\big)=0.
\end{gather}
Using~\er{umnb},~\er{gbq2},~\er{gagbpd2},~\er{ddtaq21},~\er{ddtaq22},
one obtains  
\begin{multline}
\lb{pdpdpdzcrq2}
\Big(\frac{\pd^2}{\pd u_{\ocs}\pd u_{\ocs+4}}-\frac{\pd^2}{\pd u_{\ocs+1}\pd u_{\ocs+3}}
+\frac12\frac{\pd^2}{\pd u_{\ocs+2}\pd u_{\ocs+2}}\Big)
\Big(D_x(\gb)-D_t(\ga)+[\ga,\gb]\Big)=\\
=D_x\Big(\gb_{u_{\ocs}u_{\ocs+4}}-\gb_{u_{\ocs+1}u_{\ocs+3}}+
\frac12\gb_{u_{\ocs+2}u_{\ocs+2}}\Big)
+\gb_{u_{\ocs-1}u_{\ocs+4}}-
\frac{\pd^2}{\pd u_{\ocs}\pd u_{\ocs-1}}(\ga)+\\
+\Big[\ga,\,\gb_{u_{\ocs}u_{\ocs+4}}-\gb_{u_{\ocs+1}u_{\ocs+3}}+
\frac12\gb_{u_{\ocs+2}u_{\ocs+2}}\Big]=
\\
=\frac52\bigg(D_x\Big(\frac{\pd^2}{\pd u_{\ocs}\pd u_{\ocs}}(\ga)\Big)+
\Big[\ga,\,\frac{\pd^2}{\pd u_{\ocs}\pd u_{\ocs}}(\ga)\Big]\bigg).
\end{multline}
As $D_x(\gb)-D_t(\ga)+[\ga,\gb]=0$, equation~\er{pdpdpdzcrq2} 
yields~\er{dxpd2aun} in the case $\kd=2$. 

Finally, consider the case $\kd=3$. Then $\ocs\ge 4$.

Equation~\er{utukd} reads 
$u_t=u_{7}+f(x,t,u_0,u_1,u_2,u_3,u_4,u_5)$.
According to~\er{buq2a},~\er{buq2b}, for $\kd=3$ we have  
\beq
\lb{gbq3}
\gb=u_{\ocs+6}\frac{\pd}{\pd u_\ocs}(\ga)
-u_{\ocs+1}u_{\ocs+5}\frac{\pd^2}{\pd u_{\ocs}\pd u_{\ocs}}(\ga)
+u_{\ocs+5}\gb_{01}(x,t,u_0,\dots,u_{\ocs})+\gb_{00}(x,t,u_0,\dots,u_{\ocs+4}).
\ee
Differentiating~\eqref{kdvqzcr} with respect to $u_{\ocs+5}$, $u_{\ocs+i}$ 
for $i=2,3,4$ and using~\er{gbq3}, one gets 
\begin{equation*}
-\frac{\pd^2}{\pd u_{\ocs}\pd u_{\ocs}}(\ga)+\frac{\pd^2}{\pd u_{\ocs+2}\pd u_{\ocs+4}}(\gb_{00})=0,
\qquad
\frac{\pd^2}{\pd u_{\ocs+3}\pd u_{\ocs+4}}(\gb_{00})=
\frac{\pd^2}{\pd u_{\ocs+4}\pd u_{\ocs+4}}(\gb_{00})=0.
\end{equation*}
Therefore, $\gb_{00}=\gb_{00}(x,t,u_0,\dots,u_{\ocs+4})$ is of the form 
\beq
\label{buq2bq3}
\gb_{00}=u_{\ocs+2}u_{\ocs+4}\frac{\pd^2}{\pd u_{\ocs}\pd u_{\ocs}}(\ga)
+u_{\ocs+4}\gb_{001}(x,t,u_0,\dots,u_{\ocs+1})+\gb_{000}(x,t,u_0,\dots,u_{\ocs+3})
\ee
for some $\gb_{001}(x,t,u_0,\dots,u_{\ocs+1})$ and $\gb_{000}(x,t,u_0,\dots,u_{\ocs+3})$.

Applying the operator~$\dfrac{\pd^2}{\pd u_{\ocs+3}\pd u_{\ocs+4}}$ 
to equation~\er{kdvqzcr} and using~\er{gbq3},~\er{buq2bq3}, we obtain 
\beq
\lb{gagbpd2q3}
\frac{\pd^2}{\pd u_{\ocs+3}\pd u_{\ocs+3}}(\gb_{000})+
\frac{\pd^2}{\pd u_{\ocs}\pd u_{\ocs}}(\ga)=0.
\ee
Since we consider the case when $\kd=3$ and $\ocs\ge 4$, formula~\er{dtga} implies
\begin{gather}
\lb{ddtaq31}
\frac{\pd^2}{\pd u_{\ocs}\pd u_{\ocs+6}}\big(D_t(\ga)\big)=
\frac{\pd^2}{\pd u_{\ocs}\pd u_{\ocs-1}}(\ga),\\
\lb{ddtaq32}
\Big(-\frac{\pd^2}{\pd u_{\ocs+1}\pd u_{\ocs+5}}
+\frac{\pd^2}{\pd u_{\ocs+2}\pd u_{\ocs+4}}
-\frac12\frac{\pd^2}{\pd u_{\ocs+3}\pd u_{\ocs+3}}\Big)\big(D_t(\ga)\big)=0.
\end{gather}

Using~\er{umnb},~\er{gbq3},~\er{buq2bq3},~\er{gagbpd2q3},~\er{ddtaq31},~\er{ddtaq32} 
one gets 
\begin{multline}
\lb{pdpdpdzcrq3}
\Big(\frac{\pd^2}{\pd u_{\ocs}\pd u_{\ocs+6}}-\frac{\pd^2}{\pd u_{\ocs+1}\pd u_{\ocs+5}}
+\frac{\pd^2}{\pd u_{\ocs+2}\pd u_{\ocs+4}}
-\frac12\frac{\pd^2}{\pd u_{\ocs+3}\pd u_{\ocs+3}}\Big)
\Big(D_x(\gb)-D_t(\ga)+[\ga,\gb]\Big)=\\
=D_x\Big(\gb_{u_{\ocs}u_{\ocs+6}}-\gb_{u_{\ocs+1}u_{\ocs+5}}+
\gb_{u_{\ocs+2}u_{\ocs+4}}-\frac12\gb_{u_{\ocs+3}u_{\ocs+3}}\Big)
+\gb_{u_{\ocs-1}u_{\ocs+6}}-
\frac{\pd^2}{\pd u_{\ocs}\pd u_{\ocs-1}}(\ga)+\\
+\Big[\ga,\,\gb_{u_{\ocs}u_{\ocs+6}}-\gb_{u_{\ocs+1}u_{\ocs+5}}+
\gb_{u_{\ocs+2}u_{\ocs+4}}-\frac12\gb_{u_{\ocs+3}u_{\ocs+3}}\Big]=\\
=\frac72\bigg(D_x\Big(\frac{\pd^2}{\pd u_{\ocs}\pd u_{\ocs}}(\ga)\Big)+
\Big[\ga,\,\frac{\pd^2}{\pd u_{\ocs}\pd u_{\ocs}}(\ga)\Big]\bigg).
\end{multline}
Since $D_x(\gb)-D_t(\ga)+[\ga,\gb]=0$, equation~\er{pdpdpdzcrq3} 
implies~\er{dxpd2aun} in the case $\kd=3$. 
\end{proof}
\begin{lemma}
\lb{ga2uk}
One has
\beq
\lb{ukga2}
\frac{\pd^3}{\pd u_k\pd u_{\ocs}\pd u_{\ocs}}(\ga)=0\qquad\qquad\forall\,k\in\zp.
\ee
\end{lemma}
\begin{proof}
Suppose that~\eqref{ukga2} does not hold. Let $k_0$ be the maximal integer 
such that $\dfrac{\pd^3}{\pd u_{k_0}\pd u_{\ocs}\pd u_{\ocs}}(\ga)\neq 0$. 
Equation~\er{kdvd=0} for $s=k_0+1$ says  
\beq
\lb{pduk0ga0}
\frac{\pd}{\pd u_{k_0+1}}(\ga)\,\bigg|_{u_k=a_k,\ k\ge k_0+1}=0.
\ee
Differentiating~\er{dxpd2aun} with respect to~$u_{k_0+1}$, we obtain 
\beq
\lb{pdga2g0}
\frac{\pd^3}{\pd u_{k_0}\pd u_{\ocs}\pd u_{\ocs}}(\ga)+
\Big[\frac{\pd}{\pd u_{k_0+1}}(\ga),\,\frac{\pd^2}{\pd u_{\ocs}\pd u_{\ocs}}(\ga)\Big]=0.
\ee
Substituting $u_k=a_k$ in~\er{pdga2g0} for all $k\ge k_0+1$ and 
using~\er{pduk0ga0}, one gets $\dfrac{\pd^3}{\pd u_{k_0}\pd u_{\ocs}\pd u_{\ocs}}(\ga)=0$, 
which contradicts our assumption.				 
\end{proof}

Using equation~\er{kdvd=0} for $s=\ocs$ and equation~\er{ukga2} for all $k\in\zp$, 
we see that $\ga$ is of the form
\beq
\label{ga210}
\ga=(u_{\ocs}-a_\ocs)^2\ga_2(x,t)+\ga_0(x,t,u_0,\dots,u_{\ocs-1}),
\ee
where $\ga_2(x,t)$ is a power series 
in the variables $x-x_a$, $t-t_a$ 
and $\ga_0(x,t,u_0,\dots,u_{\ocs-1})$ is a power series in the variables $x-x_a$, $t-t_a$, $u_0-a_0,\dots,u_{\ocs-1}-a_{\ocs-1}$.

From~\er{ukga2},~\er{ga210} it follows that equation~\er{dxpd2aun} reads 
\beq
\lb{pdxga2g0}
2\frac{\pd}{\pd x}(\ga_2)+2[\ga_0,\ga_2]=0.
\ee
Note that condition~\er{kdvaukak} implies 
\beq
\lb{ga0ukak}
\ga_0\Big|_{u_k=a_k,\ k\ge 0}=0.
\ee
Substituting $u_k=a_k$ in~\er{pdxga2g0} for all $k\ge 0$ and using~\er{ga0ukak}, 
we get 
\beq
\lb{pdxga20}
\frac{\pd}{\pd x}(\ga_2)=0.
\ee
Combining~\er{pdxga20} with~\er{pdxga2g0}, one obtains 
\beq
\lb{ga0ga20}
[\ga_2,\ga_0]=0.
\ee 

In view of~\er{kdvgaser},~\er{ga210}, we have 
\beq
\lb{ga0sum}
\ga_0=\sum_{l_1,l_2,i_0,\dots,i_{\ocs-1}\ge 0} 
(x-x_a)^{l_1} (t-t_a)^{l_2}(u_0-a_0)^{i_0}\dots(u_{\ocs-1}-a_{\ocs-1})^{i_{\ocs-1}}\cdot
\ga^{l_1,l_2}_{i_0\dots i_{\ocs-1}0}
\ee
According to~\er{kdvgaser},~\er{ga210},~\er{pdxga20}, one has 
\beq
\lb{ga2sum}
\ga_2=\sum_{l\ge 0}(t-t_a)^l\cdot\tilde\ga^l,\qquad\qquad 
\tilde\ga^l=\ga^{0,l}_{0\dots 02}\in\fds^\ocs(\CE,a).
\ee

Combining~\er{ga210},~\er{ga0sum},~\er{ga2sum} with 
Theorem~\ref{lemgenfdq}, we see that the elements 
\beq
\lb{tga0ga}
\tilde\ga^0,\qquad\ga^{l_1,0}_{i_0\dots i_{\ocs-1}0},\qquad 
l_1,i_0,\dots,i_{\ocs-1}\in\zp,
\ee
generate the algebra $\fds^\ocs(\CE,a)$. 
Substituting $t=t_a$ in~\er{ga0ga20} and using~\er{ga0sum},~\er{ga2sum}, one gets 
\beq
\lb{tgagal1}
\big[\tilde\ga^0,
\ga^{l_1,0}_{i_0\dots i_{\ocs-1}0}\big]=0\qquad\forall\,l_1,i_0,\dots,i_{\ocs-1}\in\zp.
\ee
Since the elements~\er{tga0ga} generate the algebra $\fds^\ocs(\CE,a)$, equation~\er{tgagal1} yields
\beq
\lb{tga0fdq}
\big[\tilde\ga^0,\,\fds^\ocs(\CE,a)\big]=0.
\ee

\begin{lemma} One has
\beq
\lb{tgalfdq}
\big[\tilde\ga^l,\,\fds^\ocs(\CE,a)\big]=0\qquad\qquad\forall\,l\in\zp.
\ee
\end{lemma}
\begin{proof}
We prove~\er{tgalfdq} by induction on $l$. 
The property $\big[\tilde\ga^0,\,\fds^\ocs(\CE,a)\big]=0$ has been obtained 
in~\er{tga0fdq}.

Let $r\in\zp$ such that $\big[\tilde\ga^l,\,\fds^\ocs(\CE,a)\big]=0$ 
for all $l\le r$. Since 
$\dfrac{\pd^{l}}{\pd t^{l}}(\ga_2)\,\bigg|_{t=t_a}=l!\cdot\tilde\ga^l$, 
we get
\beq
\lb{pdkga2}
\bigg[
\frac{\pd^{l}}{\pd t^{l}}(\ga_2)\,\bigg|_{t=t_a},\,
\frac{\pd^{m}}{\pd t^{m}}(\ga_0)\,\bigg|_{t=t_a}\bigg]=0\qquad\forall\,l\le r,
\qquad\forall\,m\in\zp.
\ee
Applying the operator $\dfrac{\pd^{r+1}}{\pd t^{r+1}}$ to equation~\er{ga0ga20}, 
substituting $t=t_a$, and using~\er{pdkga2}, one obtains 
\begin{multline*}
0=\frac{\pd^{r+1}}{\pd t^{r+1}}
\big([\ga_2,\ga_0]\big)\,\bigg|_{t=t_a}
=\sum_{k=0}^{r+1}\binom{r+1}{k}\cdot
\bigg[
\frac{\pd^{k}}{\pd t^{k}}(\ga_2)\,\bigg|_{t=t_a},\,
\frac{\pd^{r+1-k}}{\pd t^{r+1-k}}(\ga_0)\,\bigg|_{t=t_a}\bigg]=\\
=
\bigg[
\frac{\pd^{r+1}}{\pd t^{r+1}}(\ga_2)\,\bigg|_{t=t_a},\,
\ga_0\,\bigg|_{t=t_a}\bigg]=\\
=\bigg[
(r+1)!\cdot\tilde\ga^{r+1},\,
\sum_{l_1,i_0,\dots,i_{\ocs-1}} 
(x-x_a)^{l_1}(u_0-a_0)^{i_0}\dots(u_{\ocs-1}-a_{\ocs-1})^{i_{\ocs-1}}\cdot
\ga^{l_1,0}_{i_0\dots i_{\ocs-1}0}\bigg],
\end{multline*}
which implies 
\beq
\lb{tgan1ga0}
\big[\tilde\ga^{r+1},\,\ga^{l_1,0}_{i_0\dots i_{\ocs-1}0}\big]=0
\qquad\qquad\forall\,l_1,i_0,\dots,i_{\ocs-1}\in\zp.
\ee
Equation~\er{tga0fdq} yields
\beq
\lb{tga0tgan1}
\big[\tilde\ga^{0},\tilde\ga^{r+1}\big]=0.
\ee
Since the elements~\er{tga0ga} generate the algebra $\fds^\ocs(\CE,a)$, 
from~\er{tgan1ga0},~\er{tga0tgan1} 
we get 
$\big[\tilde\ga^{r+1},\,\fds^\ocs(\CE,a)\big]=0$. 
\end{proof}

\begin{theorem}
\lb{thcenter}
Let $\CE$ be the infinite prolongation of an equation of the form~\er{utukd} 
with $\kd\in\{1,2,3\}$. Let $a\in\CE$. 
For each $\ocs\in\zsp$, consider the surjective 
homomorphism $\vf_\ocs\cl\fds^\ocs(\CE,a)\to\fds^{\ocs-1}(\CE,a)$
from~\er{fdnn-1}. 

If $\ocs\ge\kd+\delta_{\kd,3}$ then  
\beq
\lb{v1v2ker}
[v_1,v_2]=0\qquad\qquad\forall\,v_1\in\ker\vf_\ocs,\qquad
\forall\,v_2\in\fds^\ocs(\CE,a).
\ee
In other words, if $\ocs\ge\kd+\delta_{\kd,3}$ then 
the kernel of~$\vf_\ocs$ is contained in the center of the Lie 
algebra $\fds^\ocs(\CE,a)$.

For each $k\in\zsp$, let 
$\psi_{k}\colon\fds^{k+\kd-1+\delta_{\kd,3}}(\CE,a)\to\fds^{\kd-1+\delta_{\kd,3}}(\CE,a)$ 
be the composition of the homomorphisms
\beq
\notag
\fds^{k+\kd-1+\delta_{\kd,3}}(\CE,a)\to\fds^{k+\kd-2+\delta_{\kd,3}}(\CE,a)\to\dots
\to\fds^{\kd+\delta_{\kd,3}}(\CE,a)\to\fds^{\kd-1+\delta_{\kd,3}}(\CE,a)
\ee
from~\er{fdnn-1}. Then   
\beq
\lb{hhhhk}
[h_1,[h_2,\dots,[h_{k-1},[h_k,h_{k+1}]]\dots]]=0\qquad\qquad\forall\,h_1,\dots,h_{k+1}\in\ker\psi_k.
\ee
In particular, the kernel of $\psi_k$ is nilpotent.
\end{theorem}
\begin{proof}
Let $\ocs\ge\kd+\delta_{\kd,3}$.
Combining~\er{ga210},~\er{buq2a},~\er{ga2sum} with the definition of the homomorphism 
$$
\vf_\ocs\cl\fds^\ocs(\CE,a)\to\fds^{\ocs-1}(\CE,a),
$$ 
we see that $\ker\vf_\ocs$ 
is equal to the ideal generated by the elements $\tilde\ga^l$, $l\in\zp$.
Then~\er{v1v2ker} follows from~\er{tgalfdq}.

So we have proved that 
the kernel of the homomorphism $\vf_\ocs\cl\fds^\ocs(\CE,a)\to\fds^{\ocs-1}(\CE,a)$ 
is contained in the center of the Lie algebra $\fds^\ocs(\CE,a)$ for any
$\ocs\ge\kd+\delta_{\kd,3}$.

Let us prove~\er{hhhhk} by induction on $k$. 
Since $\psi_{1}=\vf_{\kd+\delta_{\kd,3}}$, 
for $k=1$ property~\er{hhhhk} follows from~\er{v1v2ker}.
Let $r\in\zsp$ such that~\er{hhhhk} is valid for $k=r$. 
Then for any $h'_1,h'_2,\dots,h'_{r+2}\in\ker\psi_{r+1}$ we have 
\beq
\lb{vfhn2}
\big[\vf_{r+\kd+\delta_{\kd,3}}(h'_2),\big[\vf_{r+\kd+\delta_{\kd,3}}(h'_3),\dots,
\big[\vf_{r+\kd+\delta_{\kd,3}}(h'_r),
\big[\vf_{r+\kd+\delta_{\kd,3}}(h'_{r+1}),\vf_{r+\kd+\delta_{\kd,3}}(h'_{r+2})
\big]\big]\dots\big]\big]=0,
\ee
because $\vf_{r+\kd+\delta_{\kd,3}}(h'_i)\in\ker\psi_r$ for $i=2,3,\dots,r+2$.
Equation~\er{vfhn2} says that 
\beq
\lb{hhkervf}
\big[h'_2,\big[h'_3,\dots,\big[h'_r,
\big[h'_{r+1},h'_{r+2}\big]\big]\dots\big]\big]\in\ker\vf_{r+\kd+\delta_{\kd,3}}.
\ee
Since $\ker\vf_{r+\kd+\delta_{\kd,3}}$ is contained in the center of 
$\fds^{r+\kd+\delta_{\kd,3}}(\CE,a)$,
property~\er{hhkervf} yields 
$$
\big[h'_1,\big[h'_2,\big[h'_3,\dots,\big[h'_r,
\big[h'_{r+1},h'_{r+2}\big]\big]\dots\big]\big]\big]=0.
$$
So we have proved~\er{hhhhk} for $k=r+1$. 
Clearly, property~\er{hhhhk} implies that $\ker\psi_k$ is nilpotent.
\end{proof}

Now we prove a result which is used in Example~\ref{exintu5}.

\begin{theorem}
\lb{thu5f}
Let $\CE$ be the infinite prolongation of the equation 
\beq
\lb{utu5f}
u_t=u_5+f(x,t,u_0,u_1,u_2,u_3)
\ee
for some function $f=f(x,t,u_0,u_1,u_2,u_3)$ such that 
$\frac{\pd^3 f}{\pd u_{3}\pd u_{3}\pd u_{3}}\neq 0$. 
\textup{(}More precisely, we assume that the function 
$\frac{\pd^3 f}{\pd u_{3}\pd u_{3}\pd u_{3}}$ 
is not identically zero on any connected component of the manifold~$\CE$. 
Usually, the manifold~$\CE$ is connected, 
and then our assumption means that 
$\frac{\pd^3 f}{\pd u_{3}\pd u_{3}\pd u_{3}}$ 
is not identically zero on~$\CE$.\textup{)}

Then $\fds^1(\CE,a)=\fds^0(\CE,a)=0$ and $\fds^\ocs(\CE,a)$ 
is nilpotent for all $a\in\CE$, $\ocs>1$. 
\end{theorem}
\begin{proof}
Consider an arbitrary point $a\in\CE$ given by~\er{pointevs}. 
According to Definition~\ref{dfd}, 
the algebra $\fds^1(\CE,a)$ for equation~\er{utu5f} can be described as follows.
Consider formal power series 
\begin{gather}
\label{u5fgaser}
\ga=\sum_{l_1,l_2,i_0,i_1\ge 0} 
(x-x_a)^{l_1} (t-t_a)^{l_2}(u_0-a_0)^{i_0}(u_1-a_1)^{i_1}\cdot
\ga^{l_1,l_2}_{i_0,i_1},\\
\lb{u5fgbser}
\gb=\sum_{l_1,l_2,j_0,\dots,j_{5}\ge 0} 
(x-x_a)^{l_1} (t-t_a)^{l_2}(u_0-a_0)^{j_0}\dots(u_{5}-a_{5})^{j_{5}}\cdot
\gb^{l_1,l_2}_{j_0\dots j_{5}}
\end{gather}
satisfying 
\beq
\lb{relu5f}
\ga^{l_1,l_2}_{i_0,1}=
\ga^{l_1,l_2}_{0,0}=
\gb^{0,l_2}_{0\dots 0}=0,
\qquad\qquad l_1,l_2,i_0\in\zp.
\ee
Then $\ga^{l_1,l_2}_{i_0,i_1}$, $\gb^{l_1,l_2}_{j_0\dots j_{5}}$ 
are generators of the algebra $\fds^1(\CE,a)$, and the equation 
\beq
\label{gcu5f}
D_x(\gb)-D_t(\ga)+[\ga,\gb]=0
\ee
provides relations for these generators (in addition to relations~\er{relu5f}). 
Note that here $D_t(\ga)$ is given by formula~\er{dtga} for $\kd=2$ and $\ocs=1$,
so we have
\beq
\lb{dtga21}
D_t(\ga)=\frac{\pd}{\pd t}(\ga)+
\big(u_{5}+f(x,t,u_0,u_1,u_2,u_3)\big)\frac{\pd}{\pd u_0}(\ga)
+\big(u_{6}+D_x\big(f(x,t,u_0,u_1,u_2,u_3)\big)\big)\frac{\pd}{\pd u_1}(\ga).
\ee

Similarly to~\er{gbq2}, from~\er{gcu5f} we deduce that $\gb$ is of the form 
\beq
\lb{gbu5f}
\gb=u_{5}\frac{\pd}{\pd u_1}(\ga)
-u_{2}u_{4}\frac{\pd^2}{\pd u_{1}\pd u_{1}}(\ga)
+u_{4}\gb_{01}(x,t,u_0,u_{1})+\gb_{00}(x,t,u_0,u_1,u_2,u_{3}),
\ee
where 
$\gb_{01}(x,t,u_0,u_{1})$ is a power series 
in the variables $x-x_a$, $t-t_a$, $u_0-a_0$, $u_{1}-a_{1}$ 
and $\gb_{00}(x,t,u_0,u_1,u_2,u_{3})$ is a power series 
in the variables $x-x_a$, $t-t_a$, $u_0-a_0$, $u_1-a_1$, $u_2-a_2$, $u_3-a_3$.

Differentiating~\eqref{gcu5f} with respect to $u_{4}$, $u_{3}$ 
and using~\er{gbu5f}, we get  
\beq
\lb{gb00f}
\frac{\pd^2}{\pd u_{3}\pd u_{3}}(\gb_{00})=
\frac{\pd^2}{\pd u_{1}\pd u_{1}}(\ga)+
\frac{\pd^2 f}{\pd u_{3}\pd u_{3}}\cdot\frac{\pd}{\pd u_{1}}(\ga).
\ee
Using~\er{undxq},~\er{dtga21},~\er{gbu5f},~\er{gb00f}, one can verify that  
\begin{multline}
\lb{longu5f}
\Big(-\frac{\pd^3}{\pd u_{3}\pd u_{2}\pd u_{4}}
+\frac12\frac{\pd^3}{\pd u_{3}\pd u_{3}\pd u_{3}}\Big)
\Big(D_x(\gb)-D_t(\ga)+[\ga,\gb]\Big)=\\
=\frac12\frac{\pd^3 f}{\pd u_{3}\pd u_{3}\pd u_{3}}\cdot\bigg(
D_x\Big(\frac{\pd}{\pd u_{1}}(\ga)\Big)
+\Big[\ga,\,\frac{\pd}{\pd u_{1}}(\ga)\Big]-\frac{\pd}{\pd u_{0}}(\ga)
\bigg).
\end{multline}
Since $D_x(\gb)-D_t(\ga)+[\ga,\gb]=0$ by~\er{gcu5f}, equation~\er{longu5f} implies 
\beq
\lb{pd3fbigg}
\frac{\pd^3 f}{\pd u_{3}\pd u_{3}\pd u_{3}}\cdot\bigg(
D_x\Big(\frac{\pd}{\pd u_{1}}(\ga)\Big)
+\Big[\ga,\,\frac{\pd}{\pd u_{1}}(\ga)\Big]-\frac{\pd}{\pd u_{0}}(\ga)
\bigg)=0.
\ee
As the analytic function $\frac{\pd^3 f}{\pd u_{3}\pd u_{3}\pd u_{3}}$ 
is not identically zero on any connected component of the manifold~$\CE$, 
equation~\er{pd3fbigg} yields  
\beq
\lb{dxcomu5}
D_x\Big(\frac{\pd}{\pd u_{1}}(\ga)\Big)
+\Big[\ga,\,\frac{\pd}{\pd u_{1}}(\ga)\Big]-\frac{\pd}{\pd u_{0}}(\ga)=0.
\ee
Since $\frac{\pd}{\pd u_{2}}(\ga)=0$,
differentiating~\eqref{dxcomu5} with respect to $u_{2}$, we obtain 
$\frac{\pd^2}{\pd u_{1}\pd u_{1}}(\ga)=0$.

Recall that $\ga$ is of the form~\er{u5fgaser}. 
As $\ga^{l_1,l_2}_{i_0,1}=0$ for all $l_1,l_2,i_0\in\zp$ by~\er{relu5f}, 
equation $\frac{\pd^2}{\pd u_{1}\pd u_{1}}(\ga)=0$ yields $\frac{\pd}{\pd u_{1}}(\ga)=0$.
Combining the equation $\frac{\pd}{\pd u_{1}}(\ga)=0$ with~\er{dxcomu5}, one gets 
$\frac{\pd}{\pd u_{0}}(\ga)=0$. 

Combining the equations $\frac{\pd}{\pd u_{1}}(\ga)=\frac{\pd}{\pd u_{0}}(\ga)=0$ with~\er{relu5f}, 
we get $\ga^{l_1,l_2}_{i_0,i_1}=0$ for all $l_1,l_2,i_0,i_1\in\zp$.

Since, by Theorem~\ref{lemgenfdq}, the algebra $\fds^1(\CE,a)$ 
is generated by the elements 
$\ga^{l_1,0}_{i_0,i_1}$ for $l_1,i_0,i_1\in\zp$, we obtain $\fds^1(\CE,a)=0$. 
As one has the surjective homomorphism 
$\fds^1(\CE,a)\to\fds^0(\CE,a)$ in~\er{fdnn-1}, one gets $\fds^0(\CE,a)=0$.

According to Theorem~\ref{thcenter} for $\kd=2$, 
for any $k\in\zsp$ the kernel of the homomorphism 
$$
\psi_{k}\colon\fds^{k+1}(\CE,a)\to\fds^{1}(\CE,a)
$$ 
is nilpotent. 
Since $\fds^1(\CE,a)=0$, this implies that $\fds^\ocs(\CE,a)$ is nilpotent for all $\ocs>1$.
\end{proof}

\section{ZCRs with values in infinite-dimensional Lie algebras}
\lb{szcrinf}

According to Definition~\ref{dzcr} and Remark~\ref{rfxtu},
for a finite-dimensional Lie algebra~$\mg$,
a ZCR with values in~$\mg$ is given by analytic functions 
$A(x,t,u_0,u_1,\dots)$, $B(x,t,u_0,u_1,\dots)$ with values in~$\mg$
satisfying~\er{mnzcr}.

Sometimes one needs to consider ZCRs with values in infinite-dimensional Lie algebras.
An example of such a ZCR is studied in Section~\ref{szcrsak}.

For an arbitrary infinite-dimensional Lie algebra $\bl$, 
the notion of analytic functions with values in~$\bl$ is not defined.
Because of this, a theory for ZCRs with values in infinite-dimensional Lie algebras 
is developed below by using formal power series instead of analytic functions.

Consider an arbitrary scalar evolution equation~\er{eveq_intr}.
Let $\CE$ be the infinite prolongation of~\er{eveq_intr}.
Fix a point $a\in\CE$ given by~\er{pointevs}, which is determined by constants $x_a$, $t_a$, $a_k$.

\begin{definition}
\lb{deffzcr}
Let $\bl$ be a (possibly infinite-dimensional) Lie algebra.
A \emph{formal ZCR of order~$\le\oc$ with coefficients in $\bl$}
is given by formal power series 
\begin{gather}
\label{anasum}
\anA=\sum_{l_1,l_2,i_0,\dots,i_\ocs\ge 0} 
(x-x_a)^{l_1} (t-t_a)^{l_2}(u_0-a_0)^{i_0}\dots(u_\ocs-a_\ocs)^{i_\ocs}\cdot
\anA^{l_1,l_2}_{i_0\dots i_\ocs},\\
\label{anbsum}
\anB=
\sum_{l_1,l_2,j_0,\dots,j_{\ocs+\eo-1}\ge 0} 
(x-x_a)^{l_1} (t-t_a)^{l_2}(u_0-a_0)^{j_0}\dots(u_{\ocs+\eo-1}-a_{\ocs+\eo-1})^{j_{\ocs+\eo-1}}\cdot
\anB^{l_1,l_2}_{j_0\dots j_{\ocs+\eo-1}}
\end{gather}
such that 
\begin{gather}
\lb{abbl}
\anA^{l_1,l_2}_{i_0\dots i_\ocs},\anB^{l_1,l_2}_{j_0\dots j_{\ocs+\eo-1}}\in\bl,\\
\lb{dxanb0}
D_x(\anB)-D_t(\anA)+[\anA,\anB]=0.
\end{gather}
If the power series~\er{anasum},~\er{anbsum} 
satisfy \er{agd=0}, \er{agaukak}, \er{agbxx0} 
then this formal ZCR is said to be \emph{$a$-normal}.
\end{definition}

\begin{example}
Since \er{gasumxt}, \er{gbsumxt} 
obey \er{xgbtga}, \er{pdgau}, \er{gaua0}, \er{gbxua0} and 
$\ga^{l_1,l_2}_{i_0\dots i_\ocs},\gb^{l_1,l_2}_{j_0\dots j_{\ocs+\eo-1}}\in\fd^\oc(\CE,a)$, 
the power series \er{gasumxt}, \er{gbsumxt} constitute an $a$-normal 
formal ZCR of order~$\le\oc$ with coefficients in $\fd^\oc(\CE,a)$.
\end{example}

\begin{example}
Consider a ZCR of order~$\le\oc$ with values in a finite-dimensional Lie algebra~$\mg$ 
given by $\mg$-valued functions $A=A(x,t,u_0,\dots,u_\ocs)$, 
$B=B(x,t,u_0,\dots,u_{\ocs+\eo-1})$ satisfying~\er{mnzcr}.
If the functions $A$, $B$ are analytic on a neighborhood of the point $a\in\CE$,
then the Taylor series of these functions constitute a formal ZCR of order~$\le\oc$ with coefficients in~$\mg$.
\end{example}

For any vector space $V$, 
we denote by $\gl(V)$ the vector space of linear maps $V\to V$.
The space~$\gl(V)$ is an associative algebra with respect to the composition of such maps 
and is a Lie algebra with respect to the commutator.
We denote by $\id_V\in\gl(V)$ the identity map $\id_V\cl V\to V$.

Let $m,n\in\zp$. Consider power series
\begin{gather}
\lb{psergl}
P=\sum_{l_1,l_2,i_0,\dots,i_m\ge 0} 
(x-x_a)^{l_1} (t-t_a)^{l_2}(u_0-a_0)^{i_0}\dots(u_m-a_m)^{i_m}\cdot 
P^{l_1,l_2}_{i_0\dots i_m},\\
\notag
Q=\sum_{l_1,l_2,i_0,\dots,i_{n}\ge 0} 
(x-x_a)^{l_1} (t-t_a)^{l_2}(u_0-a_0)^{i_0}\dots(u_{n}-a_{n})^{i_{n}}\cdot 
Q^{l_1,l_2}_{i_0\dots i_{n}}
\end{gather}
with coefficients $P^{l_1,l_2}_{i_0\dots i_m},Q^{l_1,l_2}_{i_0\dots i_{n}}\in\gl(V)$.

The product $PQ$ is defined in the standard way, using the associative multiplication of the coefficients.
The power series $D_x(P)$, $D_t(P)$, $[P,Q]$ are defined as described in Remark~\ref{psdxdtlie}.
Thus $PQ$, $D_x(P)$, $D_t(P)$, $[P,Q]$ 
are power series in the variables $x-x_a$, $t-t_a$, $u_k-a_k$ with coefficients in~$\gl(V)$.

If the coefficient $P^{0,0}_{0\dots 0}\in\gl(V)$ in~\er{psergl} is invertible 
(i.e., the linear map $P^{0,0}_{0\dots 0}\cl V\to V$ is invertible), then we can consider 
the power series $P^{-1}$ such that $PP^{-1}=P^{-1}P=\id_V$.

For any Lie algebra $\bl$, 
there is a (possibly infinite-dimensional) vector space $V$ such that 
$\bl$ is isomorphic to a Lie subalgebra of~$\gl(V)$.
For example, one can use the following well-known construction. 
Denote by $\un(\bl)$ the universal enveloping algebra of $\bl$.
Using the canonical embedding $\bl\subset\un(\bl)$,
we get the injective homomorphism of Lie algebras 
\beq
\notag
\xi\cl \bl\hookrightarrow\gl(\un(\bl)),\quad\qquad \xi(v)(w)=vw,
\quad\qquad v\in \bl\subset\un(\bl),\quad\qquad w\in\un(\bl),\quad\qquad vw\in\un(\bl).
\ee
So one can set $V=\un(\bl)$.

As said above, Theorem~\ref{thnfzcr} about analytic ZCRs is proved in~\cite{scal13}.
Similarly, one can prove the following analog of Theorem~\ref{thnfzcr} for formal ZCRs.

\begin{theorem}
\lb{tnffrzcr}
Let $\oc\in\zp$. 
Consider a vector space $V$ and a Lie subalgebra $\bl\subset\gl(V)$.
Note that $V$ and $\bl$ can be infinite-dimensional.
Consider a formal ZCR of order~$\le\oc$ with coefficients in $\bl$ 
given by power series $\anA$, $\anB$ 
satisfying \er{anasum}, \er{anbsum}, \er{abbl}, \er{dxanb0}.

Then there is a unique power series of the form
\beq
\lb{gfps}
\anG=\id_V+\sum_{\substack{l_1,l_2,i_0,\dots,i_{m}\ge 0\\
l_1+l_2+i_0+\dots+i_{m}>0}} 
(x-x_a)^{l_1} (t-t_a)^{l_2}(u_0-a_0)^{i_0}\dots(u_{m}-a_{m})^{i_{m}}\cdot
\anG^{l_1,l_2}_{i_0\dots i_{m}},\qquad 
\anG^{l_1,l_2}_{i_0\dots i_{m}}\in\gl(V),
\ee
such that the power series 
\beq
\lb{ftanab}
\tilde{\anA}=\anG \anA \anG^{-1}-D_x(\anG)\cdot \anG^{-1},\qquad\qquad
\tilde{\anB}=\anG \anB \anG^{-1}-D_t(\anG)\cdot \anG^{-1}
\ee
satisfy 
\begin{gather}
\label{fd=0}
\frac{\pd\tilde{\anA}}{\pd u_s}\,\,\bigg|_{u_k=a_k,\ k\ge s}=0\qquad\quad
\forall\,s\ge 1,\\
\lb{faukak}
\tilde{\anA}\,\Big|_{u_k=a_k,\ k\ge 0}=0,\\
\lb{fbxx0}
\tilde{\anB}\,\Big|_{x=x_a,\ u_k=a_k,\ k\ge 0}=0.
\end{gather}

Furthermore, one has the following.
\begin{itemize}
\item The power series~\er{gfps} depends only on the variables
$x-x_a$, $t-t_a$, $u_k-a_k$ for $k=0,1,\dots,\oc-1$.
That is, one can write $m=\oc-1$ in~\er{gfps}.
\textup{(}In particular, if $\oc=0$ then~\er{gfps} depends only on $x-x_a$, $t-t_a$.\textup{)}
\item The power series~\er{ftanab} are of the form
\begin{gather}
\label{tanasum}
\tilde{\anA}=\sum_{l_1,l_2,i_0,\dots,i_\ocs\ge 0} 
(x-x_a)^{l_1} (t-t_a)^{l_2}(u_0-a_0)^{i_0}\dots(u_\ocs-a_\ocs)^{i_\ocs}\cdot
\tilde{\anA}^{l_1,l_2}_{i_0\dots i_\ocs},\\
\label{tanbsum}
\tilde{\anB}=
\sum_{l_1,l_2,j_0,\dots,j_{\ocs+\eo-1}\ge 0} 
(x-x_a)^{l_1} (t-t_a)^{l_2}(u_0-a_0)^{j_0}\dots(u_{\ocs+\eo-1}-a_{\ocs+\eo-1})^{j_{\ocs+\eo-1}}\cdot
\tilde{\anB}^{l_1,l_2}_{j_0\dots j_{\ocs+\eo-1}}
\end{gather}
for some $\tilde{\anA}^{l_1,l_2}_{i_0\dots i_\ocs},\tilde{\anB}^{l_1,l_2}_{j_0\dots j_{\ocs+\eo-1}}\in\bl$
and obey 
\beq
\lb{zcrtab}
D_x(\tilde{\anB})-D_t(\tilde{\anA})+[\tilde{\anA},\tilde{\anB}]=0.
\ee
That is, $\tilde{\anA}$, $\tilde{\anB}$ 
constitute a formal ZCR of order~$\le\oc$ with coefficients in~$\bl$.
Equations \er{fd=0}, \er{faukak}, \er{fbxx0} say that this ZCR is $a$-normal.
\item The power series~\er{gfps} satisfies the following.
\beq
\lb{cbl}
\text{The coefficients of the power series 
$\,\frac{\pd}{\pd x}(\anG)\cdot\anG^{-1}$, $\,\frac{\pd}{\pd t}(\anG)\cdot\anG^{-1}$, 
$\,\frac{\pd}{\pd u_k}(\anG)\cdot\anG^{-1}$, $\,k\in\zp$,
belong to~$\bl$.}
\ee
\end{itemize}
\end{theorem}

Fix a vector space $V$ and a Lie subalgebra $\bl\subset\gl(V)$.
A formal power series of the form~\er{gfps} satisfying~\er{cbl} is called a \emph{formal gauge transformation}.
It is easily seen that formal gauge transformations constitute a group with respect 
to the associative multiplication of power series with coefficients in~$\gl(V)$.
Formulas~\er{ftanab} determine an action of the group 
of formal gauge transformations on the set of formal ZCRs with coefficients in~$\bl$.

The formal ZCR given by~\er{ftanab} is gauge equivalent to the formal ZCR 
given by $\anA$, $\anB$ satisfying \er{anasum}, \er{anbsum}, \er{abbl}, \er{dxanb0}.

\begin{remark}
\lb{rhfzsr}
Equations \er{fd=0}, \er{faukak}, \er{fbxx0}, \er{zcrtab} 
imply that the following homomorphism 
\beq
\notag
\hrf\cl\fds^\ocs(\CE,a)\to\bl,\qquad
\hrf\big(\ga^{l_1,l_2}_{i_0\dots i_\ocs}\big)=
\tilde{\anA}^{l_1,l_2}_{i_0\dots i_\ocs},\qquad
\hrf\big(\gb^{l_1,l_2}_{j_0\dots j_{\ocs+\eo-1}}\big)=
\tilde{\anB}^{l_1,l_2}_{j_0\dots j_{\ocs+\eo-1}},
\ee
is well defined, where $\tilde{\anA}^{l_1,l_2}_{i_0\dots i_\ocs},\tilde{\anB}^{l_1,l_2}_{j_0\dots j_{\ocs+\eo-1}}\in\bl$ 
are the coefficients of the power series~\er{tanasum},~\er{tanbsum}.
Theorem~\ref{tnffrzcr} implies that any formal ZCR of order~$\le\oc$ 
with coefficients in~$\bl$ is gauge equivalent to an 
$a$-normal formal ZCR corresponding to a homomorphism $\hrf\cl\fds^\ocs(\CE,a)\to\bl$.

We will use this in Remark~\ref{rfdsak}, 
in order to get some information about the algebra $\fds^0(\CE,a)$ for equation~\er{esak}.
\end{remark}

\section{Integrability conditions and examples of proving non-integrability}
\lb{secicn}

\subsection{Necessary conditions for integrability}
\lb{secintegr}

As said in Section~\ref{subsint1}, in this paper, 
integrability of PDEs is understood in the sense of soliton theory
and the inverse scattering method, relying on the use of ZCRs.

For each scalar evolution equation~\er{eveq_intr},
in Section~\ref{btcsev} we have defined the family of Lie algebras $\fds^\ocs(\CE,a)$, 
where $\CE$ is the infinite prolongation of~\er{eveq_intr}, 
$a$ is a point of the manifold~$\CE$, and $\ocs\in\zp$.
In this subsection and in Subsection~\ref{sbhheq}  
we show that, using the algebras $\fds^\ocs(\CE,a)$, 
one obtains some necessary conditions for integrability of equations~\er{eveq_intr}.
Examples of the use of these conditions in 
proving non-integrability for some equations 
are given in Example~\ref{exintu5} and in Subsection~\ref{sbhheq}.

In this subsection, $\mg$ is a finite-dimensional matrix Lie algebra, and 
$\mathcal{G}$ is the connected matrix Lie group corresponding to~$\mg$.
(The precise definition of~$\mathcal{G}$ is given in Definition~\ref{dlggt}.)
A gauge transformation is a matrix-function 
$G=G(x,t,u_0,\dots,u_l)$ with values in~$\mathcal{G}$, where $l\in\zp$.
ZCRs and gauge transformations are supposed to be defined 
on a neighborhood of a point $a\in\CE$. 

As said in Section~\ref{subsint1}, all algebras are supposed to be over the field~$\fik$, 
where $\fik$ is either $\Com$ or $\mathbb{R}$, 
and the variables $x$, $t$, $u_k$ take values in~$\fik$.

\begin{definition}
\lb{dzcrgnil}
A $\mg$-valued ZCR 
$$
A=A(x,t,u_0,u_1,\dots),\qquad B=B(x,t,u_0,u_1,\dots),\qquad D_x(B)-D_t(A)+[A,B]=0
$$
is called \emph{gauge-nilpotent} if there is 
a gauge transformation $G=G(x,t,u_0,\dots,u_l)$ such that the functions 
\beq
\notag
\tilde{A}=GAG^{-1}-D_x(G)\cdot G^{-1},\qquad\qquad
\tilde{B}=GBG^{-1}-D_t(G)\cdot G^{-1}
\ee
take values in a nilpotent Lie subalgebra of $\mg$.
In other words, a $\mg$-valued ZCR is gauge-nilpotent iff it is gauge 
equivalent to a ZCR with values in a nilpotent Lie subalgebra of $\mg$.
\end{definition}

It is known that a ZCR with values in a nilpotent Lie algebra 
cannot establish integrability of a given equation~\er{eveq_intr}.
Therefore, a gauge-nilpotent ZCR cannot establish integrability of~\er{eveq_intr} either,
because a gauge-nilpotent ZCR is equivalent to a ZCR with values in 
a nilpotent Lie algebra.

Hence the property
\beq
\text{``there is $\mg$ such that equation~\er{eveq_intr} possesses 
a $\mg$-valued ZCR which is not gauge-nilpotent''}
\ee
can be regarded as a necessary condition for integrability of equation~\er{eveq_intr}.

It is shown in~\cite{scal13} that, for any $\mg$-valued ZCR of order~$\le\ocs$,  
there is a homomorphism $\hrf\cl\fds^\ocs(\CE,a)\to\mg$ 
such that this ZCR is gauge equivalent to 
a ZCR with values in the Lie subalgebra $\hrf\big(\fds^\ocs(\CE,a)\big)\subset\mg$.
The construction of~$\hrf$ is described in Remark~\ref{fdrzcr}.

Therefore, if for each $\ocs\in\zp$ and each $a\in\CE$ 
the Lie algebra $\fds^\ocs(\CE,a)$ is nilpotent then any ZCR of~\er{eveq_intr} 
is gauge-nilpotent, which implies that equation~\er{eveq_intr} is not integrable.
This yields the following result. 
\begin{theorem}
\lb{thgeninteg}
Let $\CE$ be the infinite prolongation of an equation of the form~\er{eveq_intr}.
If for each $\ocs\in\zp$ and each $a\in\CE$ 
the Lie algebra $\fds^{\ocs}(\CE,a)$ is nilpotent, 
then this equation is not integrable.

In other words, the property 
\beq
\lb{fdgennots}
\text{``there exist $\ocs\in\zp$ and $a\in\CE$ such that 
the Lie algebra $\fds^{\ocs}(\CE,a)$ is not nilpotent''}
\ee
is a necessary condition for integrability of equation~\er{eveq_intr}.
\end{theorem}

For some classes of equations~\er{eveq_intr}
one can find a nonnegative integer $r$ such that for any $k>r$ 
the algebra $\fds^k(\CE,a)$ is obtained from $\fds^{k-1}(\CE,a)$ 
by central extension.
This implies that for any $k>r$ 
the algebra $\fds^k(\CE,a)$ is obtained from $\fds^{r}(\CE,a)$ 
by applying several times the operation of central extension.
Then condition~\er{fdgennots} should be checked for $\ocs=r$.

For example, according to Theorem~\ref{thcenter} and Remark~\ref{remfdceukd}, 
for equations of the form~\er{utukd} we can take $r=\kd-1+\delta_{\kd,3}$. 
According to Proposition~\ref{fdockn}, 
for the Krichever-Novikov equation~\er{knedef} one can take $r=1$.

Let us show how this works for equations~\er{utukd}.

\begin{theorem}
\lb{fanilsol}
Let $\CE$ be the infinite prolongation of an equation of the form~\er{utukd} 
with $\kd\in\{1,2,3\}$. Let $a\in\CE$. 
If the Lie algebra $\fds^{\kd-1+\delta_{\kd,3}}(\CE,a)$ is nilpotent, 
then $\fds^\ocs(\CE,a)$ is nilpotent for all $\ocs\in\zp$. 
\end{theorem}
\begin{proof}
According to Theorem~\ref{thcenter} and Remark~\ref{remfdceukd}, 
for every $\ocs\ge\kd+\delta_{\kd,3}$ the Lie algebra $\fds^\ocs(\CE,a)$ 
is obtained from $\fds^{\kd-1+\delta_{\kd,3}}(\CE,a)$ 
by applying several times the operation of central extension. 

Since the homomorphisms~\er{fdnn-1} are surjective, for each $\tilde{\ocs}\le\kd-1+\delta_{\kd,3}$ 
we have a surjective homomorphism $\fds^{\kd-1+\delta_{\kd,3}}(\CE,a)\to\fds^{\tilde{\ocs}}(\CE,a)$.

Clearly, these properties imply the statement of Theorem~\ref{fanilsol}.
\end{proof}

Combining Theorem~\ref{fanilsol} with Theorem~\ref{thgeninteg}, we obtain the following. 
\begin{theorem}
\lb{integfq12}
Let $\CE$ be the infinite prolongation of an equation of the form~\er{utukd}
with $\kd\in\{1,2,3\}$. 

If for all $a\in\CE$ the Lie algebra $\fds^{\kd-1+\delta_{\kd,3}}(\CE,a)$ is nilpotent, 
then for each $\ocs\in\zp$ any ZCR of order~$\le\ocs$
\beq
\notag
A=A(x,t,u_0,u_1,\dots,u_\ocs),\qquad B=B(x,t,u_0,u_1,\dots),\qquad D_x(B)-D_t(A)+[A,B]=0
\ee
is gauge-nilpotent. 
Hence, if\/   
$\fds^{\kd-1+\delta_{\kd,3}}(\CE,a)$ is nilpotent for all $a\in\CE$, 
then equation~\er{utukd} is not integrable.

In other words, the property 
\beq
\lb{fdnots}
\text{``the Lie algebra $\fds^{\kd-1+\delta_{\kd,3}}(\CE,a)$ is not nilpotent 
for some $a\in\CE$''}
\ee
is a necessary condition for integrability of equations of the form~\er{utukd}.
\end{theorem}

\begin{remark}
\lb{rsymcl} 
In this paper we study integrability by means of ZCRs. 
Another well-known approach to integrability uses symmetries and conservation laws. 
Many remarkable classification results for some types of equations~\er{eveq_intr}
possessing higher-order symmetries or conservation laws are known 
(see, e.g.,~\cite{mesh-sok13,mikh91,sand-wang2009} and references therein). 

However, it is also known that the approach of symmetries and conservation laws 
is not completely universal for the study of integrability.
For a given evolution equation, 
non-existence of higher-order symmetries and conservation laws does not 
guarantee non-integrability.
For example, in~\cite{sak93} one can find a scalar evolution equation 
which is connected with KdV by a Miura-type transformation and is, therefore,
integrable, but does not possess higher-order symmetries and conservation laws.
In Subsection~\ref{szcrsak} we present this equation and a ZCR for it.

Examples of the situation when two evolution equations are connected 
by a Miura-type transformation but only one of the equations possesses higher-order symmetries
can be found also in~\cite{sok_sym88,svin-sok_nonl90}.
(In~\cite{sok_sym88,svin-sok_nonl90} Miura-type transformations are called differential substitutions.)

When we speak about symmetries and conservation laws, 
we mean the standard notions of local symmetries and conservation 
laws~\cite{mesh-sok13,mikh91,sand-wang2009}, 
which may be of arbitrarily high order with respect to the variables~$u_k$.
One can try to consider also nonlocal symmetries depending on so-called nonlocal variables
(see, e.g.,~\cite{kras-vin_nonl84,nonl89,svin-sok_nonl90} and references therein),
but the theory of nonlocal symmetries is much less developed than that of local symmetries. 
A classification result for equations of order~$2$ satisfying certain integrability conditions 
related to existence of higher-order weakly nonlocal symmetries
is presented in~\cite{svin-sok_nonl90}.
\end{remark}

\begin{remark}
\lb{rrinf}
In this subsection we study ZCRs with values in finite-dimensional Lie algebras.
Using the theory presented in Section~\ref{szcrinf},
one can show that the results of this subsection are valid also for formal ZCRs 
with coefficients in arbitrary (possibly infinite-dimensional) Lie algebras.
\end{remark}

\begin{example}
\lb{exintu5}
Consider~\er{utukd} in the case $\kd=2$. 
Let $\CE$ be the infinite prolongation of an equation of the form
\beq
\lb{utu5i}
u_t=u_5+f(x,t,u_0,u_1,u_2,u_3).
\ee 
According to Theorem~\ref{thu5f}, 
if $\dfrac{\pd^3 f}{\pd u_{3}\pd u_{3}\pd u_{3}}\neq 0$ 
then $\fds^1(\CE,a)=0$ for all $a\in\CE$. 

Combining this with Theorem~\ref{integfq12}, we get the following. 
If $\dfrac{\pd^3 f}{\pd u_{3}\pd u_{3}\pd u_{3}}\neq 0$ 
then equation~\er{utu5i} is not integrable. 
(As said above, this means that, for any $\oc\in\zp$, 
any ZCR of order~$\le\ocs$ for this equation is gauge-nilpotent.)
\end{example}

\subsection{An evolution equation related to the H\'enon-Heiles system}
\lb{sbhheq}
The following  scalar evolution equation was studied 
by A.~P.~Fordy~\cite{fordy-hh} in connection with the H\'enon-Heiles system
\beq
\lb{hheq}
u_t=u_{5}+(8\al-2\be)u_0u_{3}+(4\al-6\be)u_1u_{2}-20\al\be u_0^2u_1,
\ee
where $\al$, $\be$ are arbitrary constants.
(In~\cite{fordy-hh} these constants are denoted by $a$, $b$, but we use 
the symbol~$a$ for a different purpose.)

If $\al=\be=0$ then \er{hheq} is the linear equation $u_t=u_{5}$.
Since we intend to study nonlinear PDEs, in what follows we suppose 
that at least one of the constants $\al$, $\be$ is nonzero.
We want to determine for which values of $\al$, $\be$ equation~\er{hheq} 
is not integrable.

The following facts were noticed in~\cite{fordy-hh}.
\begin{itemize}
\item If $\al+\be=0$ then \er{hheq} is equivalent
to the Sawada-Kotera equation. (That is, if $\al+\be=0$ 
then \er{hheq} can be transformed to the Sawada-Kotera equation
by scaling of the variables. 
As said above, we assume that 
at least one of the constants $\al$, $\be$ is nonzero.)
\item If $6\al+\be=0$ then \er{hheq} is equivalent
to the $5$th-order flow in the KdV hierarchy.
\item If $16\al+\be=0$ then \er{hheq} is equivalent
to the Kaup-Kupershmidt equation.
\end{itemize}
So in the cases $\al+\be=0$, $6\al+\be=0$, $16\al+\be=0$ 
equation~\er{hheq} is equivalent to a well-known integrable equation.

Now we need to study the case
\beq
\lb{albe0}
\al+\be\neq 0,\qquad 6\al+\be\neq 0,\qquad 16\al+\be\neq 0.
\ee
As discussed in Remark~\ref{rsymcl}, 
there are several different approaches to the notion of integrability of PDEs.
According to~\cite{fordy-hh} and references therein,
in the case~\er{albe0} equation~\er{hheq} 
is not integrable in the approach of symmetries and conservation laws.
(This means that the equation does not possess higher-order symmetries 
and conservation laws.)
However, according to Remark~\ref{rsymcl}, this fact does not 
guarantee non-integrability of~\er{hheq} in some other approaches.

Let us see what the structure of the algebras $\fds^\oc(\CE,a)$ 
can say about integrability or non-integrability of equation~\er{hheq}
in the case~\er{albe0}.
Lemma~\ref{lhhfd01} is proved in~\cite{sbt18}.
\begin{lemma}[\cite{sbt18}]
\lb{lhhfd01}
Let $\CE$ be the infinite prolongation of equation~\er{hheq}. Let $a\in\CE$.
Then 
\begin{itemize}
\item the Lie algebra $\fds^1(\CE,a)$ 
is obtained from $\fds^0(\CE,a)$ by central extension,
\item if \er{albe0} holds and $\al\neq 0$,
the algebra $\fds^0(\CE,a)$ is isomorphic to the direct 
sum of the $3$-dimensional Lie algebra $\msl_2(\fik)$ 
and an abelian Lie algebra of dimension~$\le 3$,
\item if $\al=0$ and $\be\neq 0$, the Lie algebra  
$\fds^0(\CE,a)$ is nilpotent and is of dimension~$\le 6$.
\end{itemize}
\end{lemma}

Combining Lemma~\ref{lhhfd01} with Theorem~\ref{thcenter}
and Remark~\ref{remfdceukd}, we get the following.
\begin{theorem}
\lb{thhhfd}
Let $\CE$ be the infinite prolongation of equation~\er{hheq}. Let $a\in\CE$.
Then one has the following.
\begin{itemize}
\item For any $\oc\in\zp$, the kernel of the surjective homomorphism 
$\fds^\oc(\CE,a)\to\fds^0(\CE,a)$ from~\er{fdnn-1} is nilpotent.
The algebra $\fds^\oc(\CE,a)$ is obtained from the algebra $\fds^0(\CE,a)$
by applying several times the operation of central extension. 
\item If \er{albe0} holds and $\al\neq 0$, then 
$\fds^0(\CE,a)$ is isomorphic to the direct sum of $\msl_2(\fik)$ 
and an abelian Lie algebra of dimension~$\le 3$,
and for each $\oc\in\zp$ there is a surjective homomorphism 
$\fds^\oc(\CE,a)\to\msl_2(\fik)$ with nilpotent kernel.
\item If $\al=0$ and $\be\neq 0$, the Lie algebra $\fds^\oc(\CE,a)$ is nilpotent 
for all $\oc\in\zp$.
\end{itemize}
\end{theorem}

According to Theorem~\ref{thhhfd}, 
if $\al=0$ and $\be\neq 0$ then for any $\oc\in\zp$ 
the Lie algebra $\fds^\oc(\CE,a)$ for~\er{hheq} is nilpotent. 
Then Theorem~\ref{thgeninteg} implies
that equation~\er{hheq} is not integrable in the case when $\al=0$ and $\be\neq 0$.

Now it remains to study the case when \er{albe0} holds and $\al\neq 0$.
Before doing this, we need to discuss something else.
All our experience in the study of the algebras $\fds^\oc(\CE,a)$ 
for various evolution equations suggests that the following conjecture is valid.
\begin{conjecture}
\lb{snilp-new}
Let $\CE$ be the infinite prolongation of 
a \textup{(1+1)}-dimensional evolution equation~\er{eveq_intr}.
Suppose that the equation is integrable. 
Then there exist $\ocs\in\zp$ and $a\in\CE$ such that
\begin{itemize}
\item the Lie algebra $\fds^{\ocs}(\CE,a)$ is infinite-dimensional,
\item for any nilpotent ideal $\mathfrak{I}\subset\fds^{\ocs}(\CE,a)$, 
the quotient Lie algebra $\fds^{\ocs}(\CE,a)/\mathfrak{I}$ is infinite-dimensional as well.
\end{itemize}
\end{conjecture}
This conjecture is supported by the following examples.
\begin{example}
According to~\cite{scal13} for the KdV equation, the Lie algebra 
$\fds^0(\CE,a)$ is isomorphic to the direct sum 
$\msl_2(\fik[\la])\oplus\fik^3$, 
where $\fik^3$ is a $3$-dimensional abelian Lie algebra.
Hence $\msl_2(\fik[\la])$ is embedded in $\fds^0(\CE,a)$ 
as a subalgebra of codimension~$3$.

Evidently, the infinite-dimensional Lie algebra $\msl_2(\fik[\la])$ 
does not have any nontrivial nilpotent ideals.
Hence, for any nilpotent ideal $\mathfrak{I}\subset\fds^{0}(\CE,a)$,
one has $\mathfrak{I}\cap\msl_2(\fik[\la])=0$ and, therefore, $\dim\mathfrak{I}\le 3$.
This implies that Conjecture~\ref{snilp-new} is valid for the KdV equation.
\end{example}

\begin{example}
Recall that the KdV hierarchy consists of commuting flows, 
which are scalar evolution equations of orders $2k+1$ for $k\in\zsp$.
The standard $\msl_2(\fik)$-valued ZCR for the KdV hierarchy
depends polynomially on a parameter $\la$.
Therefore, this ZCR can be viewed as a ZCR with values in $\msl_2(\fik[\la])$.

It can be shown that this gives a surjective homomorphism 
$\fds^0(\CE,a)\to\msl_2(\fik[\la])$ for each equation in the hierarchy.
Since the Lie algebra $\msl_2(\fik[\la])$ is infinite-dimensional 
and does not have any nontrivial nilpotent ideals,
this implies that $\dim\fds^{0}(\CE,a)=\infty$  
and $\dim\fds^{0}(\CE,a)/\mathfrak{I}=\infty$ 
for any nilpotent ideal $\mathfrak{I}\subset\fds^{0}(\CE,a)$, 
so Conjecture~\ref{snilp-new} holds true for each equation in the KdV hierarchy.
Using similar arguments, one can show that Conjecture~\ref{snilp-new} 
is valid also for many other hierarchies of integrable evolution equations 
possessing a ZCR with a parameter.
\end{example}

\begin{example}
According to Proposition~\ref{fdockn}, 
for the Krichever-Novikov equation $\kne(e_1,e_2,e_3)$  
in the case when $e_i\neq e_j$ for all $i\neq j$, the algebra $\fds^1(\CE,a)$ 
is isomorphic to the infinite-dimensional Lie algebra $\mR_{e_1,e_2,e_3}$.
Using the basis~\er{rbas} of this algebra, it is easy to show that $\mR_{e_1,e_2,e_3}$
does not have any nontrivial nilpotent ideals.
Therefore, $\fds^1(\CE,a)$ is infinite-dimensional 
and does not have any nontrivial nilpotent ideals, which 
implies that Conjecture~\ref{snilp-new} is valid in this case.

According to~\cite{svin-sok83},
if $e_1,e_2,e_3\in\Com$ are such that $e_i=e_j$ for some $i\neq j$, 
then the Krichever-Novikov equation $\kne(e_1,e_2,e_3)$ is connected 
by a Miura-type transformation with the KdV equation.
Using this fact and the fact that Conjecture~\ref{snilp-new} 
is valid for the KdV equation, one can show that 
Conjecture~\ref{snilp-new} holds true  
for the equation $\kne(e_1,e_2,e_3)$ when $e_i=e_j$ for some $i\neq j$.
\end{example}

\begin{example}
In this paper we study scalar evolution PDEs~\er{eveq_intr}.
As said in Remark~\ref{multevol}, 
it is possible to introduce an analog of $\fds^{\ocs}(\CE,a)$ 
for multicomponent evolution PDEs~\er{uitfix}.
Therefore, one can try to check Conjecture~\ref{snilp-new} 
for multicomponent evolution PDEs.
Computations in~\cite{zcrm17} show that 
Conjecture~\ref{snilp-new} holds true for 
the Landau-Lifshitz, nonlinear Schr\"odinger equations
(which can be regarded as $2$-component evolution PDEs) 
and for a number of other multicomponent PDEs.
\end{example}

Now return to the study of equation~\er{hheq} in the case 
when \er{albe0} holds and $\al\neq 0$.
According to Theorem~\ref{thhhfd}, 
for any $a\in\CE$ and any $\oc\in\zp$ there is a surjective 
homomorphism $\psi\cl\fds^\oc(\CE,a)\to\msl_2(\fik)$ with nilpotent kernel.
Let $\mathfrak{I}\subset\fds^\oc(\CE,a)$ be the kernel of $\psi$.
Then $\mathfrak{I}$ is a nilpotent ideal of $\fds^\oc(\CE,a)$, and we have
$$
\dim\fds^\oc(\CE,a)/\mathfrak{I}=\dim\msl_2(\fik)=3.
$$
Then Conjecture~\ref{snilp-new} implies that equation~\er{hheq} 
is not integrable in this case.

\subsection{A zero-curvature representation}
\lb{szcrsak}

Consider the KdV equation
\beq
\lb{kdvxt}
u_t=u_{xxx}+uu_x,\qquad\quad u=u(x,t),
\ee
and the equation
\beq
\lb{esak}
v_{\tilde{t}}=v^3v_{\tilde{x}\tilde{x}\tilde{x}}+3v^2v_{\tilde{x}}v_{\tilde{x}\tilde{x}}
-\tilde{x}^2v_{\tilde{x}}+3\tilde{x}v,\qquad\quad v=v(\tilde{x},\tilde{t}),
\ee
where subscripts denote derivatives.
We assume that $x$, $t$, $u$, $\tilde{x}$, $\tilde{t}$, $v$ take values in~$\fik$.

According to~\cite{sak93}, equation~\er{esak} is connected with KdV~\er{kdvxt} 
by the following Miura-type transformation
\beq
\lb{mtsak}
\tilde{t}=t,\qquad\quad
\tilde{x}=u_x,\qquad\quad
v=u_{xx}.
\ee
In~\cite{sak93} the variables $\tilde{x}$ and $\tilde{t}$ are denoted by $y$ and $s$.

Using the methods of \cite{ibr_tg85,ibr_sh80,svin-sok_cl82},
it is shown in~\cite{sak93} that equation~\er{esak} does not possess higher-order symmetries and conservation laws.
(As explained in Remark~\ref{rsymcl}, 
when we speak about symmetries and conservation laws, 
we mean the standard notions of local symmetries and conservation 
laws~\cite{mesh-sok13,mikh91,sand-wang2009}.)
We are going to present a ZCR for equation~\er{esak}.

Consider the infinite-dimensional Lie algebra 
$\msl_2(\fik[\la])\cong \msl_2(\fik)\otimes_{\fik}\fik[\lambda]$ and the map
\beq
\notag
\pd_{\la}\cl\msl_2(\fik[\la])\to\msl_2(\fik[\la]),\qquad
\pd_{\la}\big(q\otimes f\big)=
q\otimes\frac{\pd f}{\pd\la},\qquad\quad
q\in\msl_2(\fik),\qquad f\in\fik[\lambda].
\ee
We set $\fik\pd_{\la}=\big\{\,c\pd_{\la}\;\big|\;c\in\fik\,\big\}$.
That is, $\fik\pd_{\la}$ is the one-dimensional vector subspace spanned by the map~$\pd_{\la}$ 
in the vector space of all $\fik$-linear maps $\msl_2(\fik[\la])\to\msl_2(\fik[\la])$.

One has the following Lie algebra structure on the vector space $\msl_2(\fik[\la])\oplus\fik\pd_{\la}$
\beq
\notag
\big[m_1+c_1\pd_{\la},\,m_2+c_2\pd_{\la}\big]=
[m_1,m_2]+c_1\pd_{\la}(m_2)-c_2\pd_{\la}(m_1),\qquad
m_1,m_2\in\msl_2(\fik[\la]),\qquad c_1,c_2\in\fik.
\ee

Consider the following functions with values in $\msl_2(\fik[\la])\oplus\fik\pd_{\la}$
\begin{gather}
\lb{asak}
A(\tilde{x},v)=
\begin{pmatrix}
0 & -\frac{1}{v}\\
\frac{\la}{6v} & 0
\end{pmatrix}
+\frac{\tilde{x}}{v}\pd_{\la},\\
\lb{bsak}
\begin{split}
B(\tilde{x},v,v_{\tilde{x}},v_{\tilde{x}\tilde{x}})=
&\begin{pmatrix}
\frac{\tilde{x}}{6} & vv_{\tilde{x}\tilde{x}}+v_{\tilde{x}}^2+\frac{\tilde{x}^2}{v}+\frac23\la \\
 \frac16(v-\la vv_{\tilde{x}\tilde{x}}-\la v_{\tilde{x}}^2)-\frac{\la\tilde{x}^2}{6v}-\frac19\la^2
& -\frac{\tilde{x}}{6}
\end{pmatrix}+\\
&+\Big(vv_{\tilde{x}}-\tilde{x}(vv_{\tilde{x}\tilde{x}}+v_{\tilde{x}}^2)-\frac{\tilde{x}^3}{v}\Big)\pd_{\la}.
\end{split}
\end{gather}
It is straightforward to check that these functions satisfy the zero-curvature condition
\beq
\lb{dtxzcr}
D_{\tilde{x}}(B)-D_{\tilde{t}}(A)+[A,B]=0, 
\ee
where 
$D_{\tilde{x}}$, $D_{\tilde{t}}$ are the total derivative operators corresponding to equation~\er{esak}.
Therefore, the functions~\er{asak},~\er{bsak} form a ZCR for equation~\er{esak}.
This ZCR takes values in the infinite-dimensional Lie algebra $\msl_2(\fik[\la])\oplus\fik\pd_{\la}$.

This ZCR for equation~\er{esak} can be obtained from the standard ZCR of the KdV equation~\er{kdvxt}
by means of the  Miura-type transformation~\er{mtsak} and a linear change of variables.

\begin{remark}
\lb{rfdsak}
Let $\CE$ be the infinite prolongation of equation~\er{esak}.
According to Definition~\ref{dipr}, $\CE$ can be identified with the space $\fik^\infty$ with the coordinates
\beq
\notag
\tilde{x},\quad \tilde{t},\quad v,\quad v_{\tilde{x}},\quad
v_{\tilde{x}\tilde{x}},\quad v_{\tilde{x}\tilde{x}\tilde{x}},
\quad v_{\tilde{x}\tilde{x}\tilde{x}\tilde{x}},\quad\dots
\ee
Then $A(\tilde{x},v)$ and $B(\tilde{x},v,v_{\tilde{x}},v_{\tilde{x}\tilde{x}})$
given by~\er{asak},~\er{bsak} are rational functions on~$\CE$ 
with values in the Lie algebra $\msl_2(\fik[\la])\oplus\fik\pd_{\la}$.

Set $V=\msl_2(\fik[\la])\oplus\fik\pd_{\la}$. 
Consider the Lie algebra $\gl(V)$ which consists of $\fik$-linear maps $V\to V$.
We have the following injective homomorphism of Lie algebras
$$
\psi\cl\msl_2(\fik[\la])\oplus\fik\pd_{\la}\hookrightarrow\gl(V),\qquad
\psi(r)(s)=[r,s],\qquad r\in\msl_2(\fik[\la])\oplus\fik\pd_{\la},\qquad
s\in V=\msl_2(\fik[\la])\oplus\fik\pd_{\la},
$$
which is the adjoint representation of $\msl_2(\fik[\la])\oplus\fik\pd_{\la}$.
Hence $\msl_2(\fik[\la])\oplus\fik\pd_{\la}$ can be regarded as a Lie subalgebra of~$\gl(V)$.

Take a point $a\in\CE$ such that $v\neq 0$ at~$a$.
Then the functions $A(\tilde{x},v)$ and $B(\tilde{x},v,v_{\tilde{x}},v_{\tilde{x}\tilde{x}})$ 
are analytic on a neighborhood of $a\in\CE$.

Taking the Taylor series of these functions, 
we get power series with coefficients in $\msl_2(\fik[\la])\oplus\fik\pd_{\la}$.
Equation~\er{dtxzcr} implies that the Taylor series 
of $A(\tilde{x},v)$, $B(\tilde{x},v,v_{\tilde{x}},v_{\tilde{x}\tilde{x}})$
constitute a formal ZCR of order~$\le 0$ with coefficients in~$\msl_2(\fik[\la])\oplus\fik\pd_{\la}$.

Using Theorem~\ref{tnffrzcr} and Remark~\ref{rhfzsr} for the Lie algebra 
$\bl=\msl_2(\fik[\la])\oplus\fik\pd_{\la}\subset\gl(V)$, 
we obtain that this formal ZCR is gauge equivalent to an 
$a$-normal formal ZCR corresponding to a homomorphism 
$\hrf\cl\fds^0(\CE,a)\to\msl_2(\fik[\la])\oplus\fik\pd_{\la}$.
(The homomorphism $\hrf$ is uniquely determined by the ZCR 
$A(\tilde{x},v)$, $B(\tilde{x},v,v_{\tilde{x}},v_{\tilde{x}\tilde{x}})$.)

Using methods of~\cite{scal13,sbt18}, one can show that 
\beq
\lb{hfd0sak}
\text{$\hrf\big(\fds^0(\CE,a)\big)$ 
contains the subalgebra $\msl_2(\fik[\la])\subset\msl_2(\fik[\la])\oplus\fik\pd_{\la}$.}
\ee

Since the infinite-dimensional Lie algebra $\msl_2(\fik[\la])$ 
is of codimension~$1$ in~$\msl_2(\fik[\la])\oplus\fik\pd_{\la}$ and
does not have any nontrivial nilpotent ideals, property~\er{hfd0sak} yields the following.
For any nilpotent ideal $\mathfrak{I}\subset\fds^{0}(\CE,a)$, 
the quotient Lie algebra $\fds^{0}(\CE,a)/\mathfrak{I}$ is infinite-dimensional.
Therefore, Conjecture~\ref{snilp-new} is valid for equation~\er{esak}.
\end{remark}

\section*{Acknowledgements}

Gianni Manno (GM) and Sergei Igonin (SI) acknowledge 
that the present research has been partially supported by the 
following projects and grants: 
\begin{itemize}
\item PRIN project 2017 ``Real and Complex Manifolds: Topology, Geometry and holomorphic dynamics''
by the Ministry of Education, University and Research (MIUR), Italy;
\item ``Connessioni proiettive, equazioni di Monge-Amp\`ere e sistemi integrabili'' by Istituto 
Nazionale di Alta Matematica (INdAM), Italy;
\item ``MIUR grant Dipartimenti di Eccellenza 2018-2022 (E11G18000350001)'', Italy; 
\item ``Finanziamento alla ricerca 2017-2018 (53\_RBA17MANGIO)'' by Politecnico di Torino;
\item ``FIR-2013 Geometria delle equazioni differenziali'' by INdAM.
\end{itemize}

GM is a member of Gruppo Nazionale 
per le Strutture Algebriche, Geometriche e le loro Applicazioni (GNSAGA) of INdAM.
The work of SI was carried out within the framework of the 
State Programme of the Ministry of Education and Science of the Russian Federation, 
project number 1.13560.2019/13.1.

SI would like to thank A.~P.~Fordy, I.~S.~Krasilshchik, A.~V.~Mikhailov, 
and V.~V.~Sokolov for useful discussions.
SI is grateful to the Max Planck Institute for Mathematics (Bonn, Germany) 
for its hospitality and excellent working conditions 
during 06.2010--09.2010, when part of this research was done.

\end{document}